\documentclass[twoside,leqno,twocolumn]{article}

\usepackage[letterpaper]{geometry}

\usepackage{ltexpprt}
\usepackage{hyperref}

\usepackage[switch]{lineno}

\usepackage{bm}
\usepackage{colortbl}
\usepackage{color}

\usepackage{hhline}
\usepackage{xspace}

\usepackage{amsmath}
\usepackage{amssymb}
\usepackage{enumitem}
\usepackage{adjustbox}
\usepackage[flushleft]{threeparttable}
\usepackage{multirow}
\usepackage{mathtools}

\usepackage{algorithm}
\usepackage{algorithmic}

\usepackage{mathtools}
\usepackage{subfigure}
\usepackage{graphicx}
\usepackage{booktabs}

\definecolor{cGreen}{RGB}{0,150,0}
\definecolor{brown}{RGB}{139,64,0}

\newcommand{\ournameLow}{{uniformly co-clustered intent modeling}\xspace}

\newcommand{\ournameBold}{{\textbf{U}niformly \textbf{C}o-\textbf{C}lustered \textbf{I}ntent {M}odeling}\xspace}
\newcommand{\ournameAbbr}{{UC$^2$I}\xspace}
\newcommand{\coreModuleName}{{Uniformly Co-Clustered Intent Modeling}\xspace}
\newcommand{\coreModuleNameLow}{{uniformly co-clustered intent modeling}\xspace}
\DeclareMathOperator{\sign}{sign}

\begin{document}

\newcommand\relatedversion{}
\renewcommand\relatedversion{\thanks{The full version of the paper can be accessed at \protect\url{https://arxiv.org/abs/1902.09310}}} 

\title{\Large Enhancing Graph Collaborative Filtering via Uniformly \\Co-Clustered Intent Modeling}

\author{Jiahao Wu\thanks{The Hong Kong Polytechnic University.} \thanks{Southern University of Science and Technology.}
\and Wenqi Fan\footnotemark[1]
\and Shengcai Liu\thanks{A*STAR, Singapore.}
\and Qijiong Liu\footnotemark[1]
\and Qing Li \footnotemark[1]
\and Ke Tang  \footnotemark[2]
}

\date{}

\maketitle


\fancyfoot[R]{\scriptsize{Copyright \textcopyright\ 2024 by SIAM\\
Unauthorized reproduction of this article is prohibited}}





\begin{abstract} Graph-based collaborative filtering has emerged as a powerful paradigm for delivering personalized recommendations. Despite their demonstrated effectiveness, these methods often neglect the underlying intents of users, which constitute a pivotal facet of comprehensive user interests. Consequently, a series of approaches have arisen to tackle this limitation by introducing independent intent representations. However, these approaches fail to capture the intricate relationships between intents of different users and the compatibility between user intents and item properties.

To remedy the above issues, we propose a novel method, named \ournameLow. 
Specifically, we devise a uniformly contrastive intent modeling module to bring together the embeddings of users with similar intents and items with similar properties. This module aims to model the nuanced relations between intents of different users and properties of different items, especially those unreachable to each other on the user-item graph. To model the compatibility between user intents and item properties, we design the user-item co-clustering module, maximizing the mutual information of co-clusters of users and items. This approach is substantiated through theoretical validation, establishing its efficacy in modeling compatibility to enhance the mutual information between user and item representations. Comprehensive experiments on various real-world datasets verify the effectiveness of the proposed framework.
\end{abstract}

\section{Introduction}

Recommender systems have become an indispensable part of our daily lives, providing personalized information filtering and aiding users in making informed decisions. The key of recommendation is to infer the preference of users towards items based on their historical interactions~\cite{he-et-al:lightgcn,wu-et-al:DcRec_cikm22,fan2023jointly}. Collaborative filtering (CF) is the most commonly used technique to model user-item interactions, grounded in the assumption that users behaving similarly tend to have similar preferences~\cite{fan2022graphTrend,fan2019graphRec,wenqi23LLMRec}. Owing to the power of modeling graph data, graph neural networks (GNNs) are recently applied to enhance CF to learn informative representations for recommendation~\cite{he-et-al:lightgcn,wu-et-al:DcRec_cikm22,fan2019graphRec,msu23LLMGraph}, referred as graph collaborative filtering.

\begin{figure}[]
\centering
{\includegraphics[width=1.0\linewidth]{{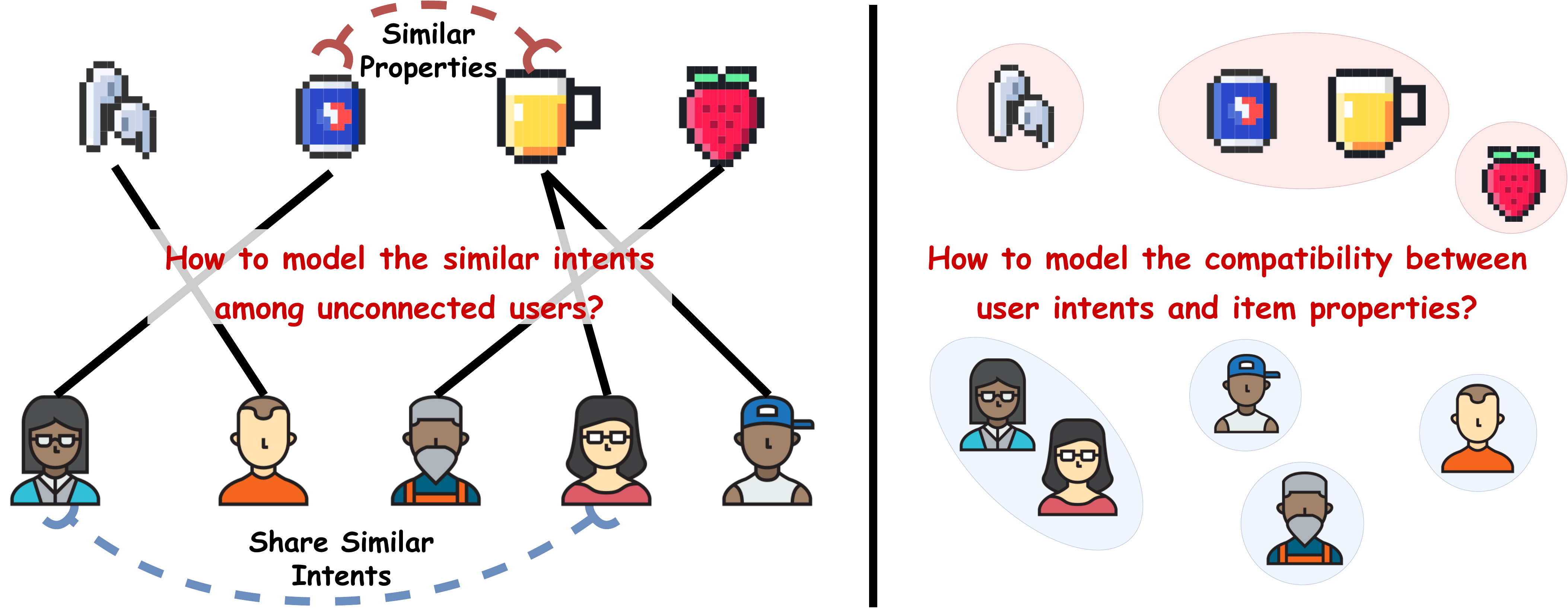}}}
\vskip -0.05in
\caption{Critical challenges in modeling user intents.}
\vskip -0.08in
\label{fig:intro}
\end{figure}

Despite the success of GNNs on modeling user preferences, prior manner of modeling user-item relationships is insufficient to capture the users' intent. In this paper, we denote a user intent as one's interests on a group of items sharing similar properties, while preference represents user's evaluation on a specific item~\cite{arXiv23intent,umap18intent}. For instance, Scarlett may choose to order beef instead of pork for dinner during fitness, although she actually prefers pork. Normally, recommending pork to Scarlett is reasonable since she may order it in the future.
Recently, some researches have been proposed to model the intents of users~\cite{nips19DisRec,umap18intent,arXiv23intent,wsdm22intent}. One of the representative methods~\cite{nips19DisRec} constructs independent representation for each intent in implicit feedback. However, several critical challenges in modeling intent still remain unsolved.

In Figure~\ref{fig:intro}, we observe that some users, inaccessible within traditional graph-based collaborative filtering, exhibit similar latent intents. However, existing methods fail to capture these affinities, leading to the first challenge: \textit{how to model the relationships between users' latent intents, especially for those not reachable in the user-item graph?} Furthermore, given that items may share similar properties, recommending items aligned with a user's past interactions or users sharing similar intents can be advantageous. Hence, we propose that explicitly encoding compatibility information in representation learning can enhance recommendation performance—a second challenge: \textit{how to model the compatibility between users' intents and item properties?}

To alleviate these two challenges, we propose a novel method, named \textbf{U}niformly \textbf{C}o-\textbf{C}lustered \textbf{I}ntent {M}odeling (\textbf{UC$^2$I}). UC$^2$I addresses the challenges via two modules: \textit{uniformly contrastive intent modeling} and \textit{user-item co-clustering}. 

To model the similarity of different user intents and different item properties, we devise a \textit{uniformly contrastive intent modeling} module, endeavoring to push the embeddings of users/items with similar intents/properties close. Specifically, we first apply clustering algorithm to the embeddings to infer intents/properties, which are the prototypes of each cluster\footnote{Therefore, in the rest of this paper, we may refer users/items by nodes or instances and refer intents/properties by prototypes.}. Then, to mitigate the negative influence of KOL (Key Opinion Leader) and popular items on the representation learning and clustering~\cite{popularityBias19FLAIRS,Wu-et-al:he2021sgl}, we pre-compute the optimal locations of prototypes to guide the learning process~\cite{Florian-et-al:dissect-cl}, as illustrated in Figure~\ref{fig:entire_frame}. Throughout this paper, we denote the optimal locations by targets. Prior to the learning process, we assign those targets to each instance based on the similarity between targets and their prototypes. 

Further, to capture the compatibility between user intents and item properties, we propose an \textit{user-item co-clustering} module. Concretely, we maximize the mutual information of user-item co-clusters. We theoretically prove that maximizing mutual information of co-clusters is capable of improving the mutual information between representations of users and items. In addition, to avoid the instances collapsing towards prototypes, we incorporate an instance-level contrastive learning module~\cite{lin-et-al:NCL,yu-et-al:XSimGCL}, enabling instance discrimination. Empirical evaluations on four real-world datasets underscore the effectiveness of our method, and the efficacy of each module is rigorously examined.

The major contributions of this paper could be summarized as follows: 
(1) We investigate a new perspective of improving graph collaborative filtering via intent modeling. Specifically, we aim at capturing the intent relations among various users and modeling their compatibility with items' properties.
(2) We propose \ournameBold (\ournameAbbr), which includes a uniformly contrastive intent modeling mechanism and a user-item co-clustering module for modeling intent-item compatibility, with theoretical validation.
(3) We conduct comprehensive experiments on several real-world datasets to validate the effectiveness of our proposed framework. Furthermore, ablation study is devised to investigate the effectiveness of the key modules.  
\section{Preliminaries}
Collaborative filtering (CF) is a fundamental technique in modeling user-item interactions. Graph collaborative filtering methods adopt graph neural networks (GNNs) as the CF backbone to obtain the representations of instances. 
Let $\mathcal{U}$ and $\mathcal{I}$ be the set of users and items respectively, while $\mathcal{O}^+ = \{y_{ui}|u\in \mathcal{U}, i\in \mathcal{I}\}$ denotes the observed interactions. In the user-item bipartite graph $\mathcal{G} = \{\mathcal{V}, \mathcal{E}\}$, the node set involves all users and items $\mathcal{V} = \mathcal{U}\cup\mathcal{I}$. The edge set is based on the observed interactions $\mathcal{E}=\mathcal{O}^+$. Typically, the representation learning of graph collaborative filtering~\cite{he-et-al:lightgcn,fan2019graphRec} consists of two steps: neighborhood aggregation and readout for the final representations, which could be formally summarized as follows:
\begin{equation}
    \begin{aligned}
        \bm{z}_{u}^{(l)}=&\;f_{combine}(\bm{z}_{u}^{(l-1)}, \;f_{aggregate}(\{\bm{z}_{v}^{(l-1)}|v\in \mathcal{N}_u\})),\\
        &\bm{z}_{u}=\;f_{readout}([\bm{z}_{u}^{(0)},\bm{z}_{u}^{(1)},...,\bm{z}_{u}^{(L)}]),
    \end{aligned}
\end{equation}
where $\mathcal{N}_u$ is the set of node $u$'s neighbors and $L$ denotes the number of GCNs layers. There are a variety of designs for the function $f_{combine}$ and $f_{aggregate}$~\cite{nips2017aggregate_2,wang2019ngcf,zgh23LLMRobust}, as well as the readout function $f_{readout}$~\cite{he-et-al:lightgcn,kdd2018readout_1,jiatong23LLMDrug}.
\section{Method}
\label{sec:framework}
In this section, we present the proposed method. The method mainly consists of three parts: (1) {Collaborative Filtering (CF) Backbone}, which aims to encode the interaction information between users and items for recommendation task; (2) {\coreModuleName} is devised to model the intents of users and compatibility between intents and items' properties, as shown in Figure~\ref{fig:entire_frame}; (3) {Instance-Level Contrastive Learning} is devised to model the instance-level information and prevent collapse of representations. 

Next, we will elaborate on each part of the proposed method and provide a theoretical justification of the user-item co-clustering module.

\subsection{Collaborative Filtering Backbone}
In this paper, we adopt LightGCN~\cite{he-et-al:lightgcn} to model the interactions between users and items. To empower the efficiency of instance-level contrastive learning, we add noises in the propagation process of each layer:
\begin{equation} 
\begin{aligned}
\bm{z}_{u}^{(l+1)} &=\sum_{i \in \mathcal{N}_{u}}( \frac{1}{\sqrt{\left|\mathcal{N}_{u}\right|\left|\mathcal{N}_{i}\right|}} \bm{z}_{i}^{(l)} + \Delta^{(l)}_i), \\
\bm{z}_{i}^{(l+1)} &=\sum_{u \in \mathcal{N}_{i}} (\frac{1}{\sqrt{\left|\mathcal{N}_{i}\right|\left|\mathcal{N}_{u}\right|}} \bm{z}^{(l)}_u + \Delta^{(l)}_u),
\end{aligned}
\label{eq:graph_prop}
\end{equation}
where $\mathcal{N}_u$ and $\mathcal{N}_i$ denote the neighbor set of user $u$ and item $i$ in the interaction graph $\mathcal{G}$. $\Delta$ is the added noise, subject to:
\begin{equation}
\Delta_{m}^{(n)}=X \odot\sign(\bm{z}_{m}^{(n)}),\;\;X \in \mathbb{R}^{d}\sim U(0,1),
\label{eq:delta}
\end{equation}

After propagating, we adopt the weighted sum function as the readout function to obtain the final representations for users and items as follows:
\begin{equation}
{\bm{z}}_{u}= \frac{1}{L+1} \sum_{l=0}^{L} {\bm{z}}_{u}^{(k)},
\;\;\;
{\bm{z}}_{i}= \frac{1}{L+1} \sum_{l=0}^{L} {\bm{z}}_{i}^{(k)},
\label{eq:layer_sum}
\end{equation}
where $L$ denotes the number of lightGCN's layers. To predict the preference of user $u$ towards item $i$, we adopt inner product:
\begin{equation}
    \hat{y}_{u,i} = \bm{z}_u^T\bm{z}_i,
    \label{eq:inner_product}
\end{equation}
where $\hat{y}_{u,i}$ is the predicted score of user $u$ and item $i$.

To optimize the primary recommendation task, we adopt Bayesian Personalized Ranking (BPR) loss~\cite{rendle-et-al:bpr2009}:
\begin{equation}
\label{eq:bpr_prediction}
    \mathcal{L}_{Rec} = \sum\limits_{(u,i,j)\in \mathcal{O}}- \log{\sigma(\hat{y}_{ui}- \hat{y}_{uj})
    },
\end{equation}
where $\mathcal{O}=\{(u,i,j)|(u,i)\in\mathcal{O}^+,(u,j)\in\mathcal{O}^-\}$. Here, $\mathcal{O^+}$ is the set of observed interactions and $\mathcal{O^-}$ is the set of unobserved ones.

\begin{figure*}[htbp]
\centering
{\includegraphics[width=0.75\linewidth]{{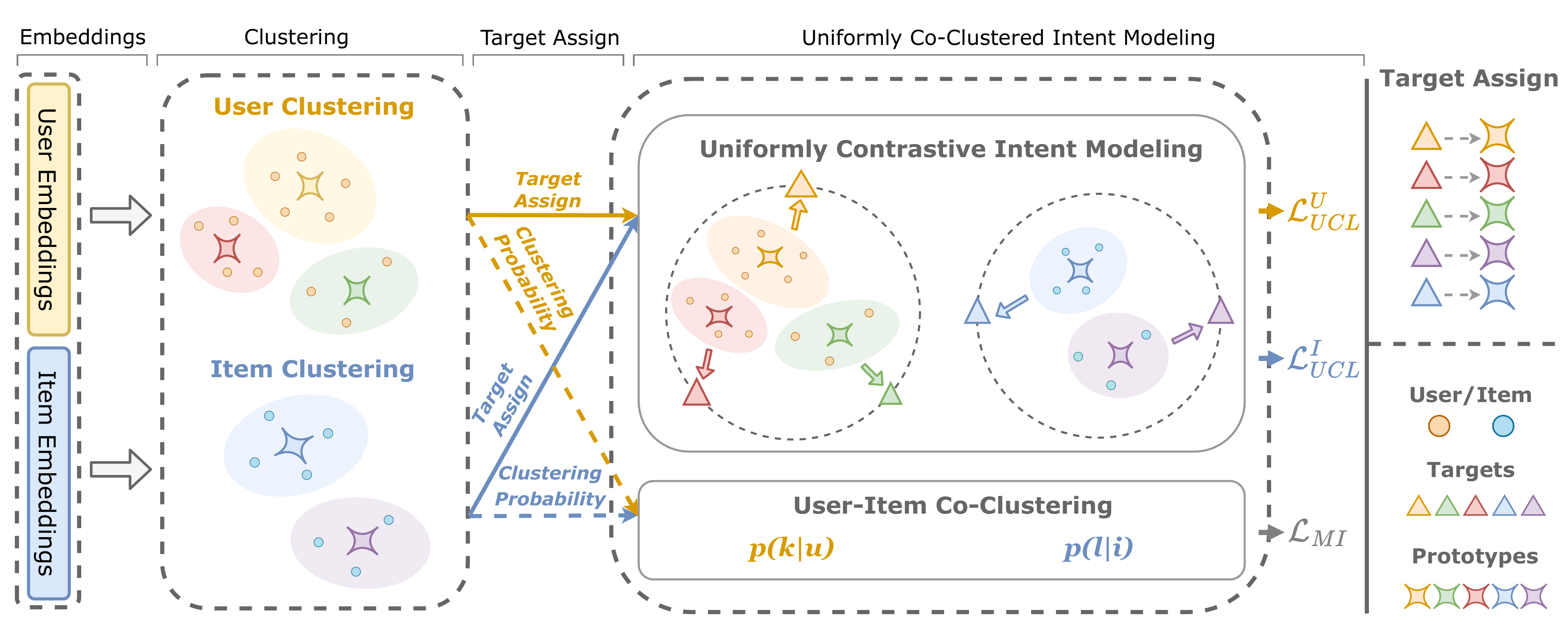}}}
\vskip -0.05in
\caption{An overview of \coreModuleName.}
\vskip -0.08in
\label{fig:entire_frame}
\end{figure*}

\subsection{\coreModuleName}
\label{subsec:uc3l}
As shown in Figure~\ref{fig:entire_frame}, the intrinsic idea of \coreModuleNameLow is two-fold: (1) We model the users' intents and items' properties via uniformly contrastive intent modeling, pushing instances towards their assigned targets; (2) To capture the co-clustered correlation information between users and items, we maximize the mutual information of their co-cluster. 

Specifically, we first pre-compute the locations of targets and apply a clustering algorithm on the embeddings of users and items to obtain their prototypes of them. Thereafter,  we assign targets to the clusters based on the similarities between their prototypes and targets. Lastly, we optimize the uniformly co-clustered objective.

\noindent \textbf{Targets Generation.} We first compute the positions for targets. Note that computing the optimal locations of targets does not require access to the data. It only requires knowing the number of prototypes and the dimension of the semantic space~\cite{Florian-et-al:dissect-cl,li-et-al:tsc_cvpr22}. Since the calculation is NP-hard, we approximate the optimal target positions of $C$ prototypes, $\{\bm{t}_i\}^C_{i=1}$, by gradient on minimizing the following loss: 
\begin{equation}
    \mathcal{L}_u\left(\left\{\bm{t}_i\right\}_{i=1}^C\right)=\frac{1}{C} \sum_{i=1}^C \log \sum_{j=1}^C e^{\bm{t}_i^T \cdot \bm{t}_j / \tau}.
    \label{eq:target_gen}
\end{equation}

\noindent\textbf{Clustering.}
We apply a clustering algorithm on the embeddings of instances (e.g., $\bm{z}^{(0)}_u$ or $\bm{z}^{(0)}_i$). Formally, the clustering process for users could be described as maximizing the following log-likelihood function:
\begin{equation}
    \sum_{u \in \mathcal{U}} \log p\left(\bm{z}^{(0)}_u \mid \Theta, \bm{R}\right)=\sum_{u \in \mathcal{U}} \log \sum_{\bm{c}_i \in C} p\left(\bm{z}^{(0)}_u, \bm{c}_i \mid \Theta, \bm{R}\right),
    \label{eq:clustering}
\end{equation}
where $\Theta$ is the set of model parameters and $\bm{R}$ is the interaction matrix between users and items. The prototype of user $u$ is denoted by $\bm{c}_i$. Symmetrically, we can define the clustering process for items.

\noindent\textbf{Uniformly Co-Clustered Intent Modeling.}

\noindent(1) \textit{Targets Assignment.} We use Hungarian Algorithm~\cite{Kuhn-et-al:assignAlgorithm} to find the assignment $\{\sigma^*_i\}^C_{i=1}$ that minimizes the distance between the prototypes and targets assigned to them. For the sake of computation efficiency, we replace the distance minimization by similarities maximization, which could be formalized as:
\begin{equation}
    \left\{\sigma_i^*\right\}_i=\underset{\left\{\sigma_i\right\}_i}{\arg \max } \frac{1}{C} \sum_{i=1}^C sim(\bm{t}_{\sigma_i}, \bm{c}_i),
\end{equation}
where $sim(\cdot,\cdot)$ is the cosine similarity function. $\bm{t}_{\sigma_i^*}$ denotes the target that is assigned to the instances whose prototype is $\bm{c}_i$. 

\noindent(2) \textit{Uniformly Contrastive Intent Modeling.} After the targets assignment, the proposed to minimize our devised uniformly contrastive learning objective as follows, which is based on InfoNCE~\cite{gutmann-et-al:2010infonce}:
\begin{equation}
\label{eq:ucl_u}
    \mathcal{L}_{UCL}^U=\sum_{u \in \mathcal{U}}-\log \frac{\exp \left(\bm{z}^{(0)}_u \cdot \bm{t}_{\sigma^*_u} / \tau\right)}{\sum_{\bm{t}_m \in U} \exp \left(\bm{z}^{(0)}_u \cdot \bm{t}_m / \tau\right)},
\end{equation}
where $U=\{\bm{t}_m\}^C_{m=1}$ is the set of pre-computed targets for users and $\bm{t}_{\sigma^*_u}$ is the target that is assigned to the prototype of the user $u$. $\tau$ is the temperature hyper-parameter. The objective on the item side is identical: 
\begin{equation}
\label{eq:ucl_i}
    \mathcal{L}_{UCL}^I=\sum_{i \in \mathcal{I}}-\log \frac{\exp \left(\bm{z}^{(0)}_i \cdot \bm{t}'_{\sigma^*_i} / \tau\right)}{\sum_{\bm{t}'_n \in U'} \exp \left(\bm{z}^{(0)}_i \cdot \bm{t}'_n / \tau\right)},
\end{equation}
where $U'=\{\bm{t}'_n\}^{C'}_{n=1}$ is the set of pre-computed targets for items and $\bm{t}'_{\sigma^*_i}$ is the target that is assigned to the prototype of the item $i$. The final objective for uniformly contrastive intent modeling is as follows:
\begin{equation}
    \mathcal{L}_{UCL} = \mathcal{L}_{UCL}^U + \alpha\mathcal{L}_{UCL}^I,
\end{equation}
where $\alpha$ is the hyper-parameter to balance the weight of above two losses.

\noindent(3) \textit{User-Item Co-Clustering.} To capture the compatibility between user
intents and item properties, we directly calculate and maximize the mutual information of the co-cluster. Following the definition of mutual information for co-cluster defined in~\cite{jing-et-al:nips22coin}, the objective function is:
\begin{equation}
    MI(K ; L)=\sum_{k, l} p(k, l) \log \frac{p(k,l)}{p(k) p(l)},
    \label{eq:mutual_info}
\end{equation}
where $k\in \{1,...,C_K\}$ is the index of user clusters and $l\in \{1,...,C_L\}$ is the index of item clusters. Furthermore, $K$ and $L$ are the random variables of user clusters and item clusters. 

Here, we denote the embeddings of user $u$ ($\bm{z}^{(0)}_u$) and item $i$ ($\bm{z}^{(0)}_i$) with $\bm{u}$ and $\bm{i}$ respectively. Since the parameters of $\bm{u}$ and $\bm{i}$ are independent, we have $p(k,l|\bm{u},\bm{i})=p(k|\bm{u})p(l|\bm{i})$, where $p(k|\bm{u})$ is the clustering probability that the prototype of user $u$ is $\bm{c}_k$ and $p(l|\bm{i})$ is the clustering probability that the prototype of item $i$ is $\bm{c}_l$. The clustering probability is implemented by the normalized inner product of the instance and its corresponding prototype, i.e., $p(k|\bm{u})=norm(\bm{u})^T norm(\bm{c}_k)$. Therefore, we calculate $p(k,l)$ as follows:
\begin{equation}
    p(k,l) = \sum_{\bm{u},\bm{i}} p(k,l|\bm{u},\bm{i}) p(\bm{u},\bm{i}),
    \label{eq:k_l_cluster_joint}
\end{equation}
where $p(\bm{u},\bm{i})$ is defined as the normalized inner product of $\bm{u}$ and $\bm{i}$: $p(\bm{u},\bm{i}) = norm(\bm{u})^T norm(\bm{i})$, if user $u$ has interacted with item $i$. Otherwise, $p(\bm{u},\bm{i})=0$. Given the joint distribution of co-cluster $p(k,l)$, we can calculate the marginal distributions:
\begin{equation}
    p(k) = \sum_l p(k,l),\;\;\;
    p(l) = \sum_k p(k,l),
    \label{eq:margin_distribution}
\end{equation}
By taking Equation~(\ref{eq:k_l_cluster_joint}) and Equation~(\ref{eq:margin_distribution}) back into Equation~(\ref{eq:mutual_info}), we can obtain the optimization objective for maximizing the mutual information of co-cluster: $\mathcal{L}_{MI} = -MI(K; L)$.

\subsection{Instance-Level Contrastive Learning}
In order to prevent the collapse of the representations, we propose to adopt an instance-level contrastive learning module. For the sake of simplicity, we adopt the cross-layer contrast~\cite{yu-et-al:XSimGCL} to implement it, whose objective function is as follows:
\begin{equation}
\label{eq:ins_u}
    \mathcal{L}_{INS}^U=\sum_{u \in \mathcal{U}}-\log \frac{\exp \left(\left(\bm{z}_u^{(k)} \cdot \bm{z}_u^{(0)} / \tau\right)\right)}{\sum_{v \in \mathcal{U}} \exp \left(\left(\bm{z}_u^{(k)} \cdot \bm{z}_v^{(0)} / \tau\right)\right)},
\end{equation}
where $\bm{z}_u^{(k)}$ is the normalized output of $k$-th layer of the collborative filtering backbone. Identically, the loss on the item side $\mathcal{L}_{INS}^I$ can also be obtained similarly.
Therefore, the overall instance-level contrastive learning loss is:
\begin{equation}
    \mathcal{L}_{INS} = \mathcal{L}^U_{INS}+\alpha\mathcal{L}^I_{INS}.
\end{equation}

\subsection{Final Objective and Training Strategy}
In this subsection, we introduce the final objective and the training strategy.

\textbf{Final Objective.} We joint train both recommendation loss and losses for intent modeling:
\begin{equation}
    \mathcal{L} = \mathcal{L}_{Rec} + \lambda_1 \mathcal{L}_{UCL} + \lambda_2 \mathcal{L}_{MI} + \lambda_3 \mathcal{L}_{INS} + \|\Theta\|_2,
    \label{eq:overall_loss}
\end{equation}
where $\lambda_1$, $\lambda_2$ and $\lambda_3$ are the hyper-parameters to balance the weight of different objectives and the regularization term. $\Theta$ is the set of collaborative filtering model parameters. 

\textbf{Training Strategy.}
At the beginning of each epoch, we fix the representations of instances (e.g., users and items) and apply kmeans algorithm over them to obtain their corresponding prototypes. Then, the model is trained by optimizing the loss presented in Equation~(\ref{eq:overall_loss}).
Furthermore, to improve the reliability of instance representations for intent modeling, we first warm-up the model by training only with the recommendation objective and instance-level contrastive learning for several epochs, formulated as follows:
\begin{equation}
    \label{eq:pre-train-loss}
    \mathcal{L}_{pre} = \mathcal{L}_{Rec} + \lambda_3\mathcal{L}_{INS} + \|\Theta\|_2
\end{equation}

\subsection{Theoretical Justification}

To study the rationale of maximizing $MI(K;L)$, we show that maximizing the mutual information of the co-cluster $K$ and $L$ is able to maximize the mutual information between representations of users and items, $MI(\bm{U};\bm{I})$~\cite{jing-et-al:nips22coin}. The proof of the theorem could be found in appendixes.

\begin{theorem}
\label{th:th1}
    The mutual information $MI(\bm{U};\bm{I})$ of bottom representations $\bm{U}$ and $\bm{I}$ is lower bounded by the mutual information of co-cluster $MI(K;L)$:
    \begin{equation}
        MI(K;L)\leq MI(\bm{U};\bm{I}).
    \end{equation}
\end{theorem}

Theorem~\ref{th:th1} reveals that the mutual information of co-cluster $K$ and $L$ acts as the lower bound for the mutual information between the representations of users and items, $MI(\bm{U};\bm{I})$~\cite{jing-et-al:nips22coin}. By maximizing the mutual information $MI(K;L)$, we can effectively maximize the mutual information between the representations of users and items.

\section{Experiments}
In this section, we evaluate the performance of \ournameAbbr. Particularly, we mainly focus on answering following questions: \textbf{RQ1}: How does \ournameAbbr perform compared with exisint state-of-the-art CF models? \textbf{RQ2}: Can the proposed modules truly boost the overall recommendation performance?

\subsection{Experimental Setup}In this subsection we introduce the datasets, baselines and evaluation metrics.

\textbf{Datasets.}
In order to examine the effectiveness of our proposed framework \ournameAbbr, we utilize several real-world datasets that are various in size, density, and domains to conduct experiments and evaluation. These datasets are: Douban~\cite{douban}, Ciao~\cite{ciao}, MovieLens-1M~\cite{ml-1m}, and Amazon-Books~\cite{ammazon-books}. The datasets' statistics are presented in Table~\ref{tab:stats_datasets}. Following the setting in~\cite{lin-et-al:NCL}, we remove instances (i.e., users or items) with fewer than 15 interactions for Amazon-Books dataset. We divide each dataset into three parts: training set, validation set, and testing set. The ratio for the number of interactions for each dataset is $8:1:1$.
\begin{table}[!ht]

\centering
\label{tab:stats_datasets}
\scalebox{0.75}{
\begin{tabular}{ccccc}
\toprule
\textbf{Datasets} & \textbf{\#Users} & \textbf{\#Items} & \textbf{\#Interactions} & \textbf{Density} \\ \midrule
Douban& 2,848& 39,586& 894,887& 0.00794\\
Ciao& 7,375& 105,114& 284,086& 0.00036 \\
MovieLens-1M& 6,040& 3,629& 836,478& 0.03816\\
Amazon-Books& 58,145& 58,052& 2,517,437& 0.00075\\
\bottomrule
\end{tabular}
}
\caption{Statistics of the datasets}
\end{table}

\textbf{Baselines.} We compare the proposed model with the following baselines.
     (1) {BPRMF}~\cite{rendle-et-al:bpr2009} is a method based on matrix factorization (MF) by optimizing the BPR loss to learn the representations of users and items. 
     (2) {NeuMF}~\cite{he-et-al:NeuMF09he} proposes to replace the dot product in MF model with neural networks to learn the preference of users towards items.

         (3) {NGCF}~\cite{wang2019ngcf} models user-item interactions as a bipartite graph and utilizes GNNs to enhance CF methods.  
     
     (4) {LightGCN}~\cite{he-et-al:lightgcn} proposes to remove the redundant parts of GCN to improve its efficiency of it for recommendation modeling and achieve promising performance. 
     (5) {SGL}~\cite{Wu-et-al:he2021sgl} is the first work that introduces self-supervised learning for the graph to the recommendation.
     (6) {NCL}~\cite{lin-et-al:NCL} is an efficient prototypical contrastive learning framework for recommendation.
     (7) {XSimGCL}~\cite{yu-et-al:XSimGCL} is a graph-augmentation free contrastive learning method for recommendation. 
     
\textbf{Evaluation Metrics.}
To evaluate the performance of top-N recommendation, we follow the setting in~\cite{lin-et-al:NCL} to adopt two widely-used metrics: Recall@N and NDCG@N. Following~\cite{lin-et-al:NCL}, we set N to 10, 20, and 50 for consistency.

\textbf{Implementation Details.} Our proposed framework \ournameAbbr is implemented based on an open source framework \textsc{RecBole}\footnote{https://github.com/RUCAIBox/RecBole/tree/master}. As for the baseline NCL~\cite{lin-et-al:NCL}, we utilize the implementation provided in the original paper. The rest of the baselines are implemented on top of the open source library \textsc{QRec}\footnote{https://github.com/Coder-Yu/QRec}. Since we strictly follow settings in~\cite{lin-et-al:NCL}, the experimental results of some baselines (i.e., BPRMF, NeuMF, NGCF, LightGCN, SGL and NCL) on MovieLens-1M, and Amazon-Books datasets are referred from the performance reported in~\cite{lin-et-al:NCL}. For a fair comparison, we follow~\cite{lin-et-al:NCL} to adopt the Adam optimizer, and the embedding size is set to 64. All the experiments are run with Pytorch (version 1.12) on Nvidia A30 GPU (Memory: 24GB, Cuda version: 11.3).

\subsection{Overall Comparison.}

\begin{table*}[!ht]
\centering
\begin{adjustbox}{max width=2.0\columnwidth}
\begin{threeparttable}
\centering

\label{tab:exp-main}

\begin{tabular}{@{}c|c|cc|cc|ccc|c@{}}
\toprule[1.2pt]
Dataset & Metric & BPRMF & NeuMF  & NGCF  & LightGCN & SGL & NCL &XSimGCL&\ournameAbbr  \\ \midrule \midrule
 \multirow{6}{*}{Douban} & Recall@10 & 0.0408 & 0.0407 &  0.0493  & 0.0573 & 0.0652 & {0.0721}&\underline{0.0732} &\textbf{0.0753}\\
 & NDCG@10 & 0.1216 & 0.1230 &  0.1304  & 0.1270 & 0.1310 & {0.1316}&\underline{0.1345} &\textbf{0.1360}\\
 & Recall@20 & 0.0690 & 0.0734 &  0.0840  & 0.0878 & 0.0982 & \underline{0.1143}&0.1123 &\textbf{0.1165}\\
 & NDCG@20 & 0.1197 & 0.1226 &  {0.1298}  & 0.1280 & 0.1293 & 0.1289&\underline{0.1300} &\textbf{0.1322}\\
 & Recall@50 & 0.1354 & 0.1417 &  0.1596  & 0.1590 & 0.162 & {0.1955}&\underline{0.1979} &\textbf{0.2005}\\
 & NDCG@50 & 0.1277 & 0.1317 &  0.1401  & 0.1448 & \underline{0.1449} & 0.1444&0.1440 &\textbf{0.1477}\\\midrule
 
\multirow{6}{*}{Ciao} & Recall@10 & 0.0152& 0.0231  &  0.0237  & 0.0306 & 0.0323 & {0.0352}& \underline{0.0378}&\textbf{0.0397}\\
 & NDCG@10   & 0.0127& 0.0220 &  0.0222  & 0.0232 & {0.0244} & 0.0236 &\underline{0.0258}&\textbf{0.0259}\\
 & Recall@20 & 0.0247 & 0.0375 &  0.0391  & 0.0448 & 0.0452 & {0.0528} &\underline{0.0569}&\textbf{0.0581}\\
 & NDCG@20   & 0.0157& 0.0264  &  0.0267  & 0.0281 & \underline{0.0311} & 0.0285 &\underline{0.0311}&\textbf{0.0315}\\
 & Recall@50 & 0.0430 & 0.0637  &  0.0685  & 0.0746 & 0.0621 & {0.0822} &\underline{0.0836}&\textbf{0.0857}\\
 & NDCG@50   & 0.0211& 0.0341  &  0.0354  & {0.0379} & \underline{0.0373} &  0.0354 &0.0361&\textbf{0.0382}\\ \midrule

\multirow{6}{*}{MovieLens-1M} & Recall@10 & {0.1804} & {0.1657} & {0.1846}  & {0.1876} & {0.1888} & {0.2057} &\underline{0.2084}&\textbf{0.2098}\\
 & NDCG@10 & {0.2463} & {0.2295} & {0.2528}  & 0.2514 & 0.2526 & {0.2732} &\underline{0.2766}&\textbf{0.2768}\\
 & Recall@20 & {0.2714} & {0.2520} &  {0.2741}  & 0.2796 & 0.2848  & 0.3037 &\underline{0.3063}&\textbf{0.3106}\\
 & NDCG@20 & {0.2569} & {0.2400} &  {0.2614}  & 0.2620  & 0.2649 & {0.2843} &\underline{0.2889}&\textbf{0.2894}\\
 & Recall@50 & {0.4300} & {0.4122} & {0.4341}  & 0.4469 & 0.4487 & {0.4686}&\underline{0.4760} &\textbf{0.4769}\\
 & NDCG@50 & {0.3014} & {0.2851} & {0.3055}  & 0.3091 & 0.3111  & {0.3300}&\underline{0.3350} &\textbf{0.3356}\\ \midrule

\multirow{6}{*}{Amazon-Books} & Recall@10 & 0.0607 & 0.0507 &  0.0617 & 0.0797 & 0.0898 & {0.0933}&\underline{0.0976} &\textbf{0.0989}\\
 & NDCG@10   & 0.0430  &  0.0351 & 0.0427  & 0.0565 & 0.0645 & {0.0679}&\underline{0.0704} &\textbf{0.0714}\\
 & Recall@20 & 0.0956 &  0.0823 &  0.0978  & 0.1206 & 0.1331 & {0.1381}&\underline{0.1431} &\textbf{0.1444}\\
 & NDCG@20   & 0.0537 &  0.0447 & 0.0537  & 0.0689 & 0.0777 & {0.0815}&\underline{0.0842} &\textbf{0.0854}\\
 & Recall@50 & 0.1681 &  0.1447 &  0.1699  & 0.2012 & 0.2157 & {0.2175}&\underline{0.2251} &\textbf{0.2269}\\
 & NDCG@50   & 0.0726 &  0.0610  & 0.0725 & 0.0899 & 0.0992 & {0.1024}&\underline{0.1059} &\textbf{0.1071}\\ 
\bottomrule[1.2pt]
\end{tabular}

\end{threeparttable}
\end{adjustbox}
\caption{Overall Performance Comparison}
\vskip -0.1in
\vskip -0.05in
\end{table*}

To evaluate and verify the effectiveness of the proposed \ournameAbbr, we conduct extensive comparison experiments between it and three types of baselines: MF-based methods (i.e., BPRMF~\cite{rendle-et-al:bpr2009} and NeuMF~\cite{he-et-al:NeuMF09he}), GNNs-based methods (i.e., NGCF~\cite{wang2019ngcf} and LightGCN~\cite{he-et-al:lightgcn}) and CL-based methods (i.e., SGL~\cite{Wu-et-al:he2021sgl} and NCL~\cite{lin-et-al:NCL}). The experiment results are elaborated in Table~\ref{tab:exp-main}. We have following observations: (1) Across all the metrics on these five datasets, the GNN-based models, such as NGCF, consistently outperform the traditional MF-based models. This improvement mainly originates from the ability of GNNs-based models in encoding high-order information. Promising performance of LightGCN verifies the decency of simplified GNNs structure, indicating that redundant parameters of GNNs could lead to degradation of performance~\cite{he-et-al:lightgcn}. 

(2) As for our proposed framework, we can see that it consistently outperforms all the baselines, answering \textbf{RQ1}. Compared with previous proposed ssl-based methods, the decent performance of \ournameAbbr comes from the accurate intent modeling, which is capable of revealing the intents of users and capturing the compatibility between users' intents and items' properties.

\subsection{Ablation Study} 
In this subsection, we conduct an ablation study to investigate the effectiveness of each module of \ournameAbbr. 
We compare \ournameAbbr with its two variants: (1) {w/o UCL} denotes the variant without the uniformly contrastive intent modeling and (2) {w/o MI} denotes the variant without the user-item co-clustering module.

From the results shown in Figure~\ref{fig:performance-ablation}, \ournameAbbr outperforms these two variants. Specifically, eliminating the uniformly contrastive intent modeling module could significantly undermine performance. Under most metrics, the benefit of directly adding uniformly contrastive learning modules to \ournameAbbr is much larger than the negative effects of undermining the correlation between representations of users and items. This is also verified by the results presented in Figure~\ref{fig:performance-ablation}, where the variant that removes the user-item co-clustering module performs much better than the variant without uniformly contrastive intent modeling, answering \textbf{RQ2}. 

\begin{figure}[t]
\centering
\includegraphics[width=0.9\columnwidth]{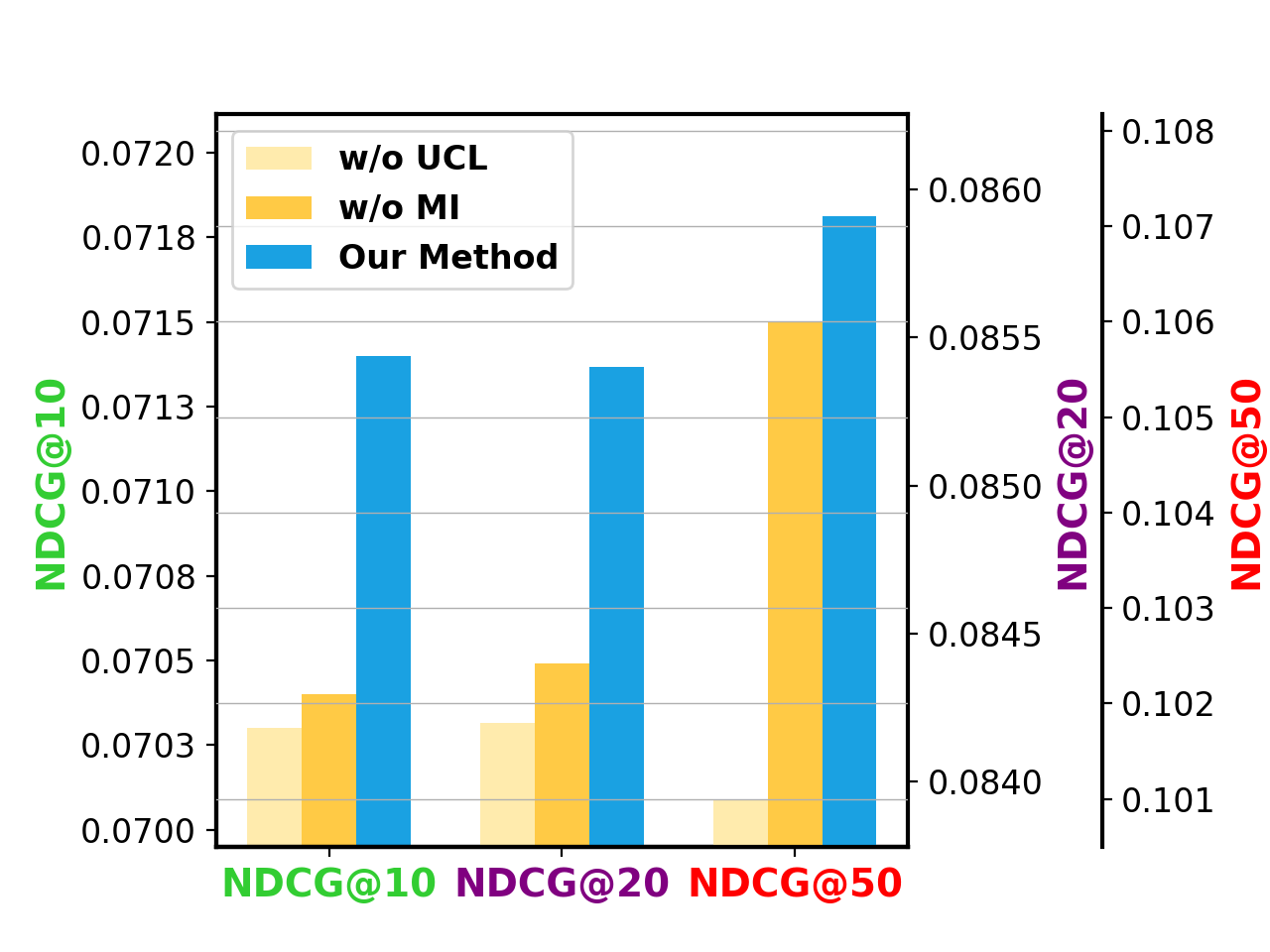}
\vskip -0.15in
\caption{Ablation study on the effectiveness of different modules. \textit{w/o UCL} denotes the variant without the uniformly contrastive intent modeling module and \textit{w/o MI} denotes the one without co-clustering module. This figure shows the results on dataset Amazon-Books and results on other datasets could be found in the appendix.}
\label{fig:performance-ablation}
\vskip -0.05in
\end{figure}
\subsection{Further Analyasis}
In this subsection, we study the effect of different model hyper-parameters. Specifically,  for the architecture of the framework, we investigate how the noise rate $\Delta$ and the number of intents influence the performance of our proposed \ournameAbbr. For simplicity, we set the same number of prototypes for the users and items. Besides, we study the loss weights (i.e., $\lambda_1$, $\lambda_2$ and $\lambda_3$). Additionally, we also explore the performance with different numbers of layers and the temperature parameter $\tau$. 

\textit{Effect of the number of intents.} To study how the uniformly contrastive learning objective affects the model's performance, we set the range of the number of prototypes to {500, 1000, 1500, 2000, 2500}. As shown in Table~\ref{tab:ablation_prototypes_num}, We can find that setting the value to 1000 could achieve the best results on most datasets. However, the variety of the performance with different numbers of prototypes does not lead to a large variance in performance. We argue that this is due to the advantage of incorporating semantic structure to connect instances that are not physically connected in the user-item interaction graph. 

\textit{Effect of the noise rate $\Delta$.} Here, we study how varying the value of noise rate influences the recommendation performance of our proposed framework. The results are shown in Table~\ref{tab:ablation_noise_rate}. 
We observe that too large a value of noise rate could lead to performance decreasing since it deteriorates the quality of the learned representations. As stated in~\cite{yu-et-al:XSimGCL}, a proper selection for the magnitude of noise rate does not affect the semantics of representations and it could serve as a kind of augmentation method for contrastive learning.

\textit{Effect of the loss weights.} We study the influence of different values of the loss weights  $\lambda_1$, $\lambda_2$ and $\lambda_3$ for contrastive learning objectives and user-item co-clustering objectives. The results are shown in Figure~\ref{fig:hyperparameter-lambdas}, from which we could find that properly selecting values for these weights is significant. This verifies the necessity of these three modules. 

\textit{Effect of temperature parameter $\tau$.} In previous researches~\cite{simclr_v1} on contrastive learning, the temperature parameter $\tau$ defined in Equations~(\ref{eq:ucl_u},\ref{eq:ucl_i},\ref{eq:ins_u}) plays a vital role in affecting performance on the downstream tasks. To explore the effect of $\tau$, we present the performances of the proposed method under different values of $\tau$ in Figure~\ref{fig:hyperparameter-temperature-tau}. 

\textit{Effect of the number of layers.} We vary the number of layer in the range of $\{$2,3,4,5$\}$. Since we adopt cross-layer contrast to implement instance-level contrastive learning, setting the number of layers to 1 would be invalid. In Table~\ref{tab:ablation-num-layers}, we report the performance of LightGCN, \ournameAbbr and the variant of \ournameAbbr that eliminates the module of contrastive learning and user-item co-clustering.

\begin{table}[]
\centering
\label{tab:ablation-num-layers}
\scalebox{0.7}{
\begin{tabular}{cccccc}
\toprule[1.2pt]
\multicolumn{2}{c}{Datasets}             & \multicolumn{2}{c}{Ciao}          & \multicolumn{2}{c}{Amazon-Books}  \\ \hline
\multicolumn{2}{c}{Metrics}              & Recall@10       & NDCG@10         & Recall@10       & NDCG@10         \\ \midrule\midrule
\multirow{3}{*}{2 Layers} & LightGCN      & 0.0341          & 0.0232          & 0.0793          & 0.0564          \\
                          & INS-CL & 0.0359          & 0.0246          & 0.0962          & 0.0697          \\
                          & \ournameAbbr        & \textbf{0.0370}  & \textbf{0.0255} & \textbf{0.0963} & \textbf{0.0700}   \\ \hline
\multirow{3}{*}{3 Layers} & LightGCN      & 0.0359          & 0.0248          & 0.0771          & 0.0551          \\
                          & INS-CL & 0.0378          & 0.0258          & 0.0976          & 0.0704          \\
                          & \ournameAbbr        & \textbf{0.0397} & \textbf{0.0264} & \textbf{0.0989} & \textbf{0.0714} \\ \hline
\multirow{3}{*}{4 Layers} & LightGCN      & 0.0349          & 0.0239          & 0.0743          & 0.0531          \\
                          & INS-CL & 0.0371          & 0.0251          & 0.0972          & 0.0697          \\
                          & \ournameAbbr        & \textbf{0.0379} & \textbf{0.0257} & \textbf{0.0975} & \textbf{0.0701} \\ \hline
\multirow{3}{*}{5 Layers} & LightGCN      & 0.0252          & 0.0167          & 0.0706          & 0.0505          \\
                          & INS-CL & 0.0383          & 0.0259          & 0.0962          & 0.0695          \\
                          & \ournameAbbr        & \textbf{0.0387} & \textbf{0.0259} & \textbf{0.0963} & \textbf{0.0695} \\ \bottomrule[1.2pt]
\end{tabular}

}
\caption{Performance comparison of different layers on Ciao and Amazon-Books over Recall@10 and NDCG@10. INS-CL denote the variant of \ournameAbbr without  uniformly contrastive intent modeling module and co-clustering module.}
\end{table}

\begin{figure}[t]
\centering
\includegraphics[width=0.7\columnwidth]{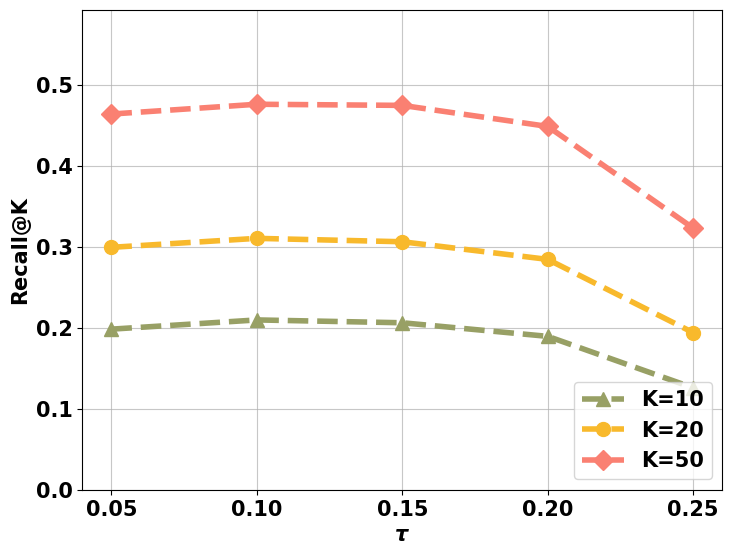}
\vskip -0.1in
\caption{Performance comparison w.r.t. different values of temperature hyper-parameter $\tau$ on dataset MovieLens-1M under metric Recall@K. The results on other datasets could be found in the appendix.}
\label{fig:hyperparameter-temperature-tau}
\vskip -0.1in
\end{figure}

\begin{table}[]
\centering
\scalebox{0.65}{
\begin{tabular}{cccccc}
\toprule[1.2pt]
Dataset                       & N. of I. & Recall@10       & Recall@20       & NDCG@10         & NDCG@20          \\ \midrule\midrule
\multirow{5}{*}{Douban}       & 500                  & 0.0740          & 0.1133          & 0.1353          & 0.1307  \\
                              & 1000                 & \textbf{0.0753} & \textbf{0.1165} & \textbf{0.1360} & \textbf{0.1322} \\
                              & 1500                 & 0.0733          & 0.1140          & 0.1341          & 0.1301  \\
                              & 2000                 & 0.0731          & 0.1151          & 0.1356          & 0.1318\\
                              & 2500                 & 0.0732          & 0.1134          & 0.1341          & 0.1299 \\ \hline
\multirow{5}{*}{Ciao}         & 500                  & 0.0377          & 0.0576          & 0.0254          & 0.0309         \\
                              & 1000                 & \textbf{0.0397} & \textbf{0.0580} & \textbf{0.0264} & \textbf{0.0316}\\
                              & 1500                 & 0.0373          & 0.0578          & 0.0253          & 0.0309          \\
                              & 2000                 & 0.0378          & \textbf{0.0580} & 0.0254          & 0.0309         \\
                              & 2500                 & 0.0372          & 0.0579          & 0.0252          & 0.0309          \\ \hline
\multirow{5}{*}{MovieLens-1M} & 500                  & 0.2087          & 0.3102          & \textbf{0.2769} & 0.2887          \\
                              & 1000                 & 0.2094          & 0.3103          & 0.2767          & 0.2887          \\
                              & 1500                 & \textbf{0.2098} & \textbf{0.3106} & 0.2761          & 0.2887          \\
                              & 2000                 & 0.2088          & 0.3101          & 0.2768          & \textbf{0.2889} \\
                              & 2500                 & 0.2085          & 0.3097          & 0.2766          & 0.2885          \\ \hline
\multirow{5}{*}{Amazon-Books} & 500                  & 0.0971          & 0.1423          & 0.0703          & 0.0841          \\
                              & 1000                 & \textbf{0.0989} & \textbf{0.1444} & \textbf{0.0714} & \textbf{0.0854}\\
                              & 1500                 & 0.0972          & 0.1428          & 0.0703          & 0.0842          \\
                              & 2000                 & 0.0972          & 0.1428          & 0.0703          & 0.0842          \\
                              & 2500                 & 0.0973          & 0.1429          & 0.0704          & 0.0843         \\ \bottomrule[1.2pt]
\end{tabular}
}
\caption{Effect of different numbers of intents on the performance. \textit{N. of I.} denotes the number of intents.}
\label{tab:ablation_prototypes_num}
\end{table}

\begin{table}[]
\centering
\scalebox{0.65}{
\begin{tabular}{cccccc}
\toprule[1.2pt]
Dataset                        & Noise Rate &  Recall@10 & Recall@20        &  NDCG@10 &  NDCG@20\\ \midrule\midrule
                               & 2.00E-01   & 0.0463          & 0.1373          & 0.0904                               & 0.0866 \\
                               & 2.00E-02   & 0.0579  & 0.1622          & 0.1073                               & 0.1036 \\
                               & 2.00E-03   & \textbf{0.0753}& \textbf{0.2005} & \textbf{0.136}                       & \textbf{0.1322}\\
                               & 2.00E-04   & 0.059  & 0.165           & 0.114                                & 0.1103\\
\multirow{-5}{*}{Douban}       & 2.00E-05   & 0.0725  & 0.199           & 0.135                                & 0.1315 \\ \hline
                               & 2.00E-01   & 0.0387 & \textbf{0.0941} & \textbf{0.0264}                      & \textbf{0.0316}\\
                               & 2.00E-02   & 0.0329  & 0.076           & 0.0231                               & 0.0275 \\
                               & 2.00E-03   & \textbf{0.0397}   & 0.0883          & 0.0264                               & 0.0316  \\
                               & 2.00E-04   & 0.0373    & 0.0884          & 0.0253                               & 0.0309  \\
\multirow{-5}{*}{Ciao}         & 2.00E-05   & 0.0366    & 0.0877          & 0.0248                               & 0.0306  \\ \hline
                               & 2.00E-01   & 0.1656    & 0.4071          & 0.2299                               & 0.2394    \\
                               & 2.00E-02   & 0.1737          & 0.4329          & 0.2401                               & 0.2509  \\
                               & 2.00E-03   & \textbf{0.2098}                        & \textbf{0.3106} & \textbf{0.2761}                      & \textbf{0.2887}\\
                               & 2.00E-04   & 0.0058                                 & 0.0111  & 0.0087                               & 0.0100 \\
\multirow{-5}{*}{MovieLens-1M} & 2.00E-05   & 0.0053                                 & 0.0102 & 0.0081                               & 0.0093  \\ \hline
                               & 2.00E-01   & 0.0959                                 & 0.1416& 0.069                                & 0.083\\
                               & 2.00E-02   & 0.0971                                 & 0.1423 & 0.0704                               & 0.0842\\
                               & 2.00E-03   & \textbf{0.0989}                        & \textbf{0.1444} & \textbf{0.0714}                      & \textbf{0.0854} \\
                               & 2.00E-04   & 0.0943                                 & 0.1398  & 0.0683                               & 0.0822  \\
\multirow{-5}{*}{Amazon-Books} & 2.00E-05   & 0.0944                                 & 0.1391  & 0.0686                               & 0.0822 \\ \bottomrule[1.2pt]
\end{tabular}
}
\caption{Performance comparison  w.r.t. different noise rate.}
\vskip -0.1in

\label{tab:ablation_noise_rate}
\end{table}

\begin{figure*}[]
\centering
\subfigure[Douban - $\lambda_1$]{\label{fig:lambda1-target-ssl_douban_recall}\includegraphics[width=0.63\columnwidth]{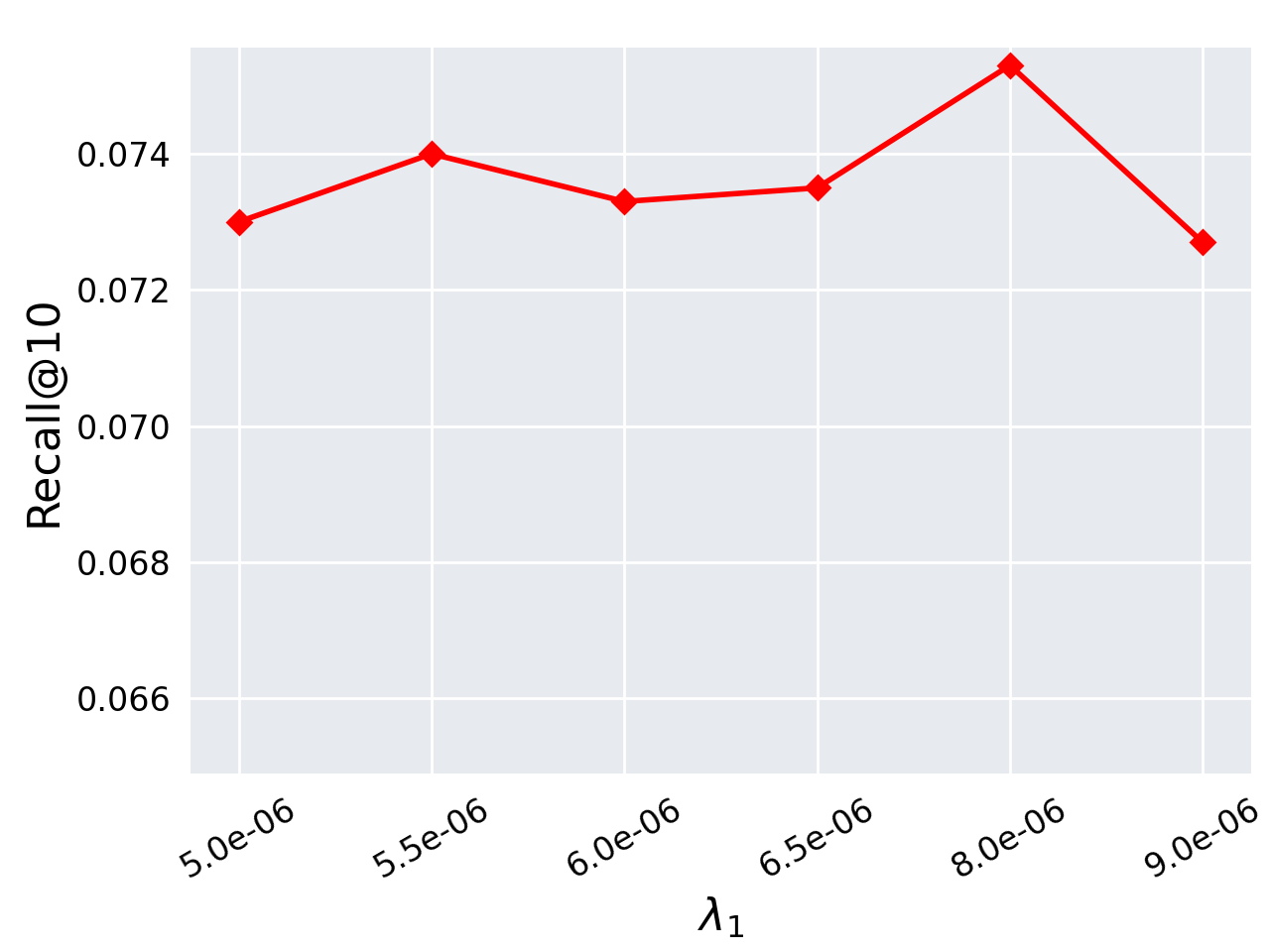}}
\subfigure[Douban - $\lambda_2$]{\label{fig:lambda2-cluster-mi_douban_recall}\includegraphics[width=0.63\columnwidth]{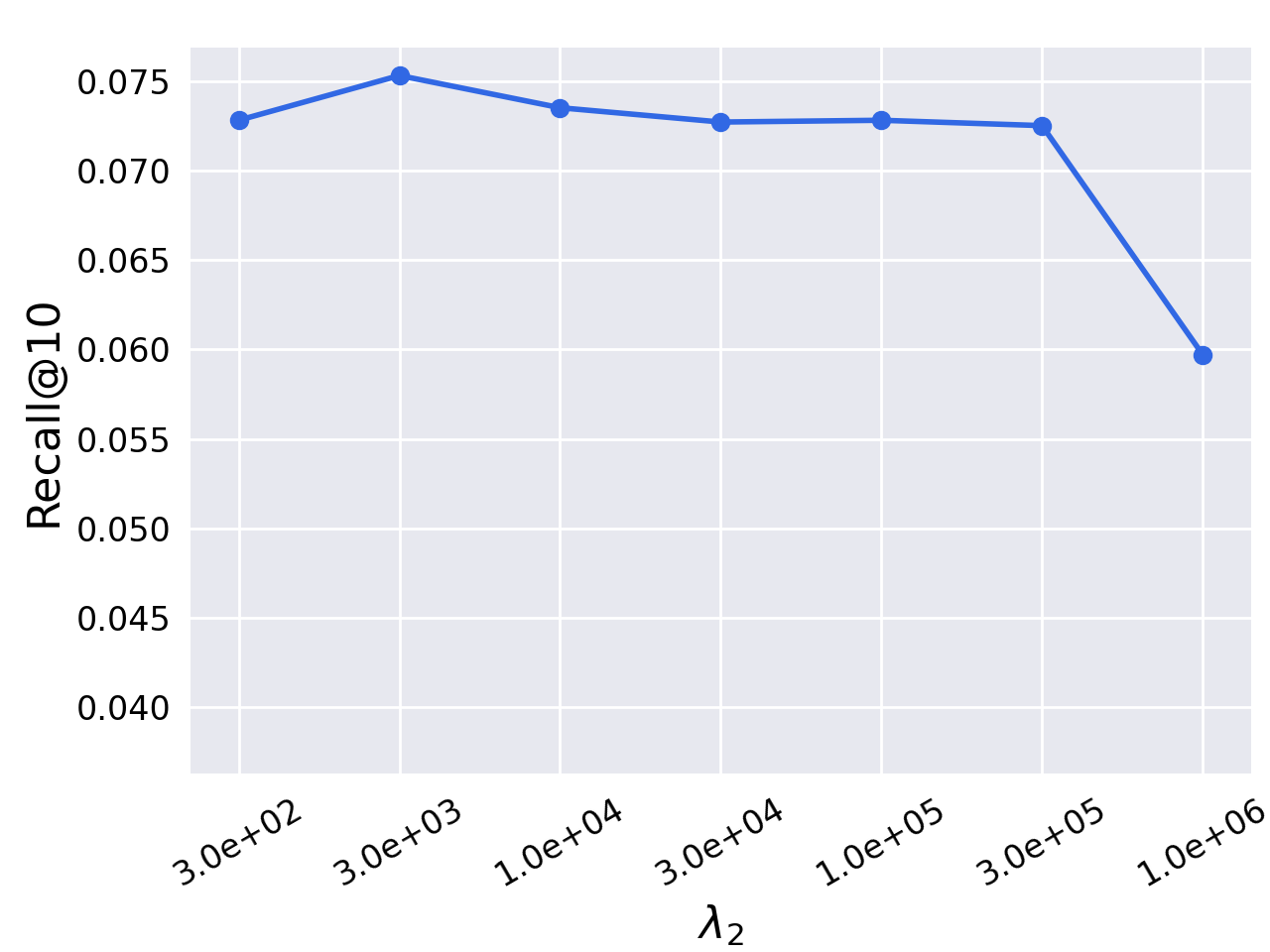}}
\subfigure[Douban - $\lambda_3$]{\label{fig:lambda3-ssl-reg_douban_recall}\includegraphics[width=0.63\columnwidth]{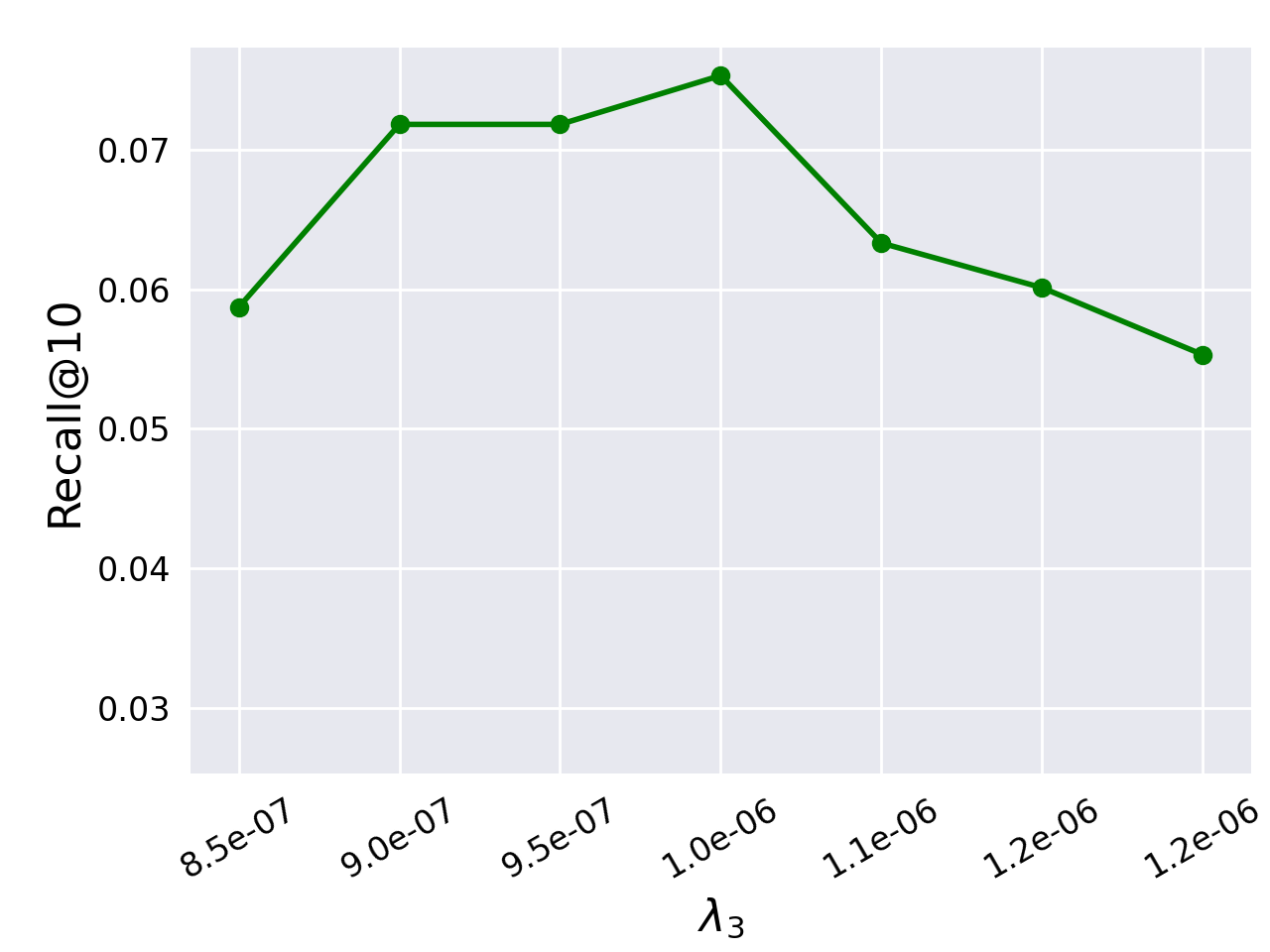}}

\vskip -0.15in
\caption{The effect of hyper-parameter $\lambda_1$, $\lambda_2$ and $\lambda_3$ under metric Recall@10 on dataset Douban. Results on other datasets could be found in the appendix.}
\label{fig:hyperparameter-lambdas}
\vskip -0.15in
\end{figure*}

\section{Related Works}
Over the past few years, graph neural networks (GNNs)~\cite{fan2019graphRec,ieeeComp23shengcai,wu2023contgcn} have exhibited excellent ability in modeling graph data and have been applied to put forward the frontier of a wide range of fields, including recommendation~\cite{liu2022prec,qijiong23www,wenjing23arXiv}. A variety of GNNs-based recommendation models (e.g., graphRec~\cite{fan2019graphRec}, NGCF~\cite{wang2019ngcf}, lightGCN~\cite{he-et-al:lightgcn}) have been proposed to explore the application of GNNs in recommendation tasks and achieved promising performance. Among those GNNs-based models, the models that originate from graph convolution networks (GCNs)~\cite{kipf2016gcn,ruihe23ITETCI} are the most popular. The GCNs-based recommendation methods model user-item interactions as bipartite graphs and then aggregate the neighborhood information in the graph layer by layer to enrich the representation of nodes, fulfilling graph reasoning~\cite{yu-et-al:XSimGCL}. NGCF~\cite{wang2019ngcf} follows the design of GCN~\cite{kipf2016gcn} and adapts it to the recommendation scenario, obtaining state-of-the-art performance. To further improve the effectiveness of GCNs on recommendation tasks, LightGCN is proposed~\cite{he-et-al:lightgcn} and it abandons unnecessary operations and parameters in GCN. Currently, LightGCN is one of the most prevalent variants of GCNs in recommendation due to its concise design and decent performance. Due to its effectiveness, LightGCN has been adopted as a collaborative filtering backbone in tons of follow-up models as well as the framework proposed in this paper.
Recent emergence of contrastive learning (CL) has brought promising improvement among various fields~\cite{cikm22sslRec, simclr_v1, jiajun2023iclr, jiajun22icml,wenqi23diffusionSurvey}. Therefore, a wave of research has been proposed to explore the effectiveness of CL in recommendation~\cite{wu-et-al:DcRec_cikm22,Wu-et-al:he2021sgl,yu-et-al:XSimGCL,lin-et-al:NCL}. The general paradigm of CL-based recommendation models devise a CL-based pretext task and then joint train it with the primary recommendation task. To mitigate the problem caused by the imbalanced-degree distribution and representations' vulnerability to noisy interactions, SGL~\cite{Wu-et-al:he2021sgl} proposes an instance-level contrastive learning framework for recommendation task, which is the first attempt of self-supervised learning in recommendation. 

Despite their decent performance, those methods either overlook the significance of modeling intent relations and compatibility between the intents and items' properties. Our proposed \ournameAbbr can well address these problems.

\section{Conclusion}
To model the users' intents and compatibility between them and items' properties, we propose a novel method, named uniformly co-clustered contrastive intent modeling, consisting of uniformly contrastive intent modeling and user-item co-clustering. Specifically, the uniformly contrastive intent modeling pushes the representations with similar intents/properties towards corresponding targets. Thereafter, to model the compatibility between the intents and properties, we derive the user-item co-clustering module to maximize the mutual information between users and items. Theoretical justification for the design of the user-item co-clustering module is provided. Extensive experiments on several real-world datasets demonstrate the effectiveness of the proposed method.


\newpage
\clearpage

\bibliographystyle{sdm}
\bibliography{references}
\appendix

\section{Proof for User-Item Co-Clustering Module}
In this section, we provide the proof for theorem~\ref{th:th1}. Before proving the theorem, we would like to introduce lemma~\ref{lemma:lemma1}.
\begin{lemma}
    \label{lemma:lemma1}
    While the bottom representations of of users (e.g., $\bm{u}\in\bm{U}$) and items (e.g., $\bm{i}\in\bm{I}$) are independent, the following inequality holds:
    \begin{equation}
        \label{ineq:lemm1}
        \log p(k) p(l) \geq \sum_{\bm{u}, \bm{i}} p(\bm{u}, \bm{i} \mid k, l) \log \frac{p(\bm{u}, k) p(\bm{i}, l)}{p(\bm{u}, \bm{i} \mid k, l)},
    \end{equation}
    where $k\in \{1,...,C_K\}$, $l\in \{1,...,C_L\}$, $\bm{u}\in\bm{U}$ and $\bm{i}\in\bm{I}$.
\end{lemma}
\begin{proof}
\label{pf:lemma_proof}
    Since $p(k) = \sum_{\bm{u}}p(\bm{u}, k)$ and $p(l)=\sum_{\bm{i}}p(\bm{i}, l)$, we can obtain $p(k)p(l) = \sum_{\bm{u}, \bm{i}}p(\bm{u},k)p(\bm{i},l)$. Therefore, according to the Jensen's inequality, we have:
    \begin{equation}
        \begin{aligned}
            \log p(k)p(l)& = \log \sum_{\bm{u}, \bm{i}}p(\bm{u},k)p(\bm{i},l) \frac{p(\bm{u}, \bm{i} \mid k, l)}{p(\bm{u}, \bm{i} \mid k, l)}\\ &\geq \sum_{\bm{u}, \bm{i}} p(\bm{u}, \bm{i} \mid k, l) \log \frac{p(\bm{u},k)p(\bm{i},l)}{p(\bm{u}, \bm{i} \mid k, l)}.
        \end{aligned}
    \end{equation}
    Therefore, the inequality~\ref{ineq:lemm1} holds.
\end{proof}
    Then, we would prove theorem~\ref{th:th1} as follows.
\begin{theorem}
    The mutual information $MI(\bm{U};\bm{I})$ of bottom representations $\bm{U}$ and $\bm{I}$ is lower bounded by the mutual information of co-cluster $MI(K;L)$:
    \begin{equation}
        MI(K;L)\leq MI(\bm{U};\bm{I}).
    \end{equation}
\end{theorem}
\begin{proof}
    According to the definition of mutual information, we calculate the mutual information of co-cluster as follows:
    \begin{equation}
        \label{eq:mi_kl_proof_1}
        \begin{aligned}
            MI(K ; L)&=\sum_{k, l} p(k, l) \log \frac{p(k,l)}{p(k) p(l)}\\ &= \sum_{k,l} p(k,l)(\log p(k,l)-log(p(k)p(l))).
        \end{aligned}
    \end{equation}
    Taking the inequality in Lemma~\ref{lemma:lemma1} into the Equation~(\ref{eq:mi_kl_proof_1}), we obtain:
    \begin{equation}
    \label{mi_kl_proof_2}
    \resizebox{.98\hsize}{!}{
        $\begin{aligned}
            MI(K;L)&\leq \sum_{k,l} p(k,l)(\log p(k,l)- \sum_{\bm{u}, \bm{i}} p(\bm{u}, \bm{i} \mid k, l) \log \frac{p(\bm{u},k)p(\bm{i},l)}{p(\bm{u}, \bm{i} \mid k, l)})\\
            &=\sum_{k,l} p(k,l)\sum_{\bm{u}, \bm{i}} p(\bm{u}, \bm{i} \mid k, l)(\log p(k,l)- \log \frac{p(\bm{u},k)p(\bm{i},l)}{p(\bm{u}, \bm{i} \mid k, l)})\\
            &=\sum_{k,l} p(k,l)\sum_{\bm{u}, \bm{i}} p(\bm{u}, \bm{i} \mid k, l)  \log \frac{p(k,l)p(\bm{u}, \bm{i} \mid k, l)}{p(\bm{u},k)p(\bm{i},l)}\\
            &=\sum_{\bm{u}, \bm{i}, k,l} p(\bm{u}, \bm{i}, k, l) \log\frac{p(\bm{u}, \bm{i}, k, l)}{p(\bm{u},k)p(\bm{i},l)}\\
            &= \sum_{\bm{u}, \bm{i}, k,l} p(\bm{u}, \bm{i}, k, l) \log\frac{p(k, l \mid \bm{u}, \bm{i}) p(\bm{u}, \bm{i})}{p(\bm{u})p(k\mid \bm{u})p(\bm{i})p(l\mid \bm{i})} = R.
        \end{aligned}$
    }
    \end{equation}
    Since the bottom representations of users ($\bm{U}$) and items ($\bm{I}$) are independent, we have $p(k, l \mid \bm{u}, \bm{i}) = p(k\mid \bm{u})p(l\mid \bm{i})$. By taking it back into $R$, we can obtain:
    \begin{equation}
    \label{eq:mi_ui}
        \begin{aligned}
            R&=\sum_{\bm{u}, \bm{i}, k,l} p(\bm{u}, \bm{i}, k, l) \log\frac{p(\bm{u}, \bm{i})}{p(\bm{u})p(\bm{i})}\\&=\sum_{\bm{u}, \bm{i}} p(\bm{u}, \bm{i}) \log\frac{p(\bm{u}, \bm{i})}{p(\bm{u})p(\bm{i})}=MI(\bm{U};\bm{I}).
        \end{aligned}
    \end{equation}
    Therefore, we prove that $MI(K;L)\leq MI(\bm{U};\bm{I})$.
\end{proof}

\section{Additional Experimental Results}
In this section, we present additional results as follows:
\begin{figure*}[htbp]
\centering
\subfigure[Douban - NDCG@K]{\label{fig:module_ablation_douban_ndcg}\includegraphics[width=0.65\columnwidth]{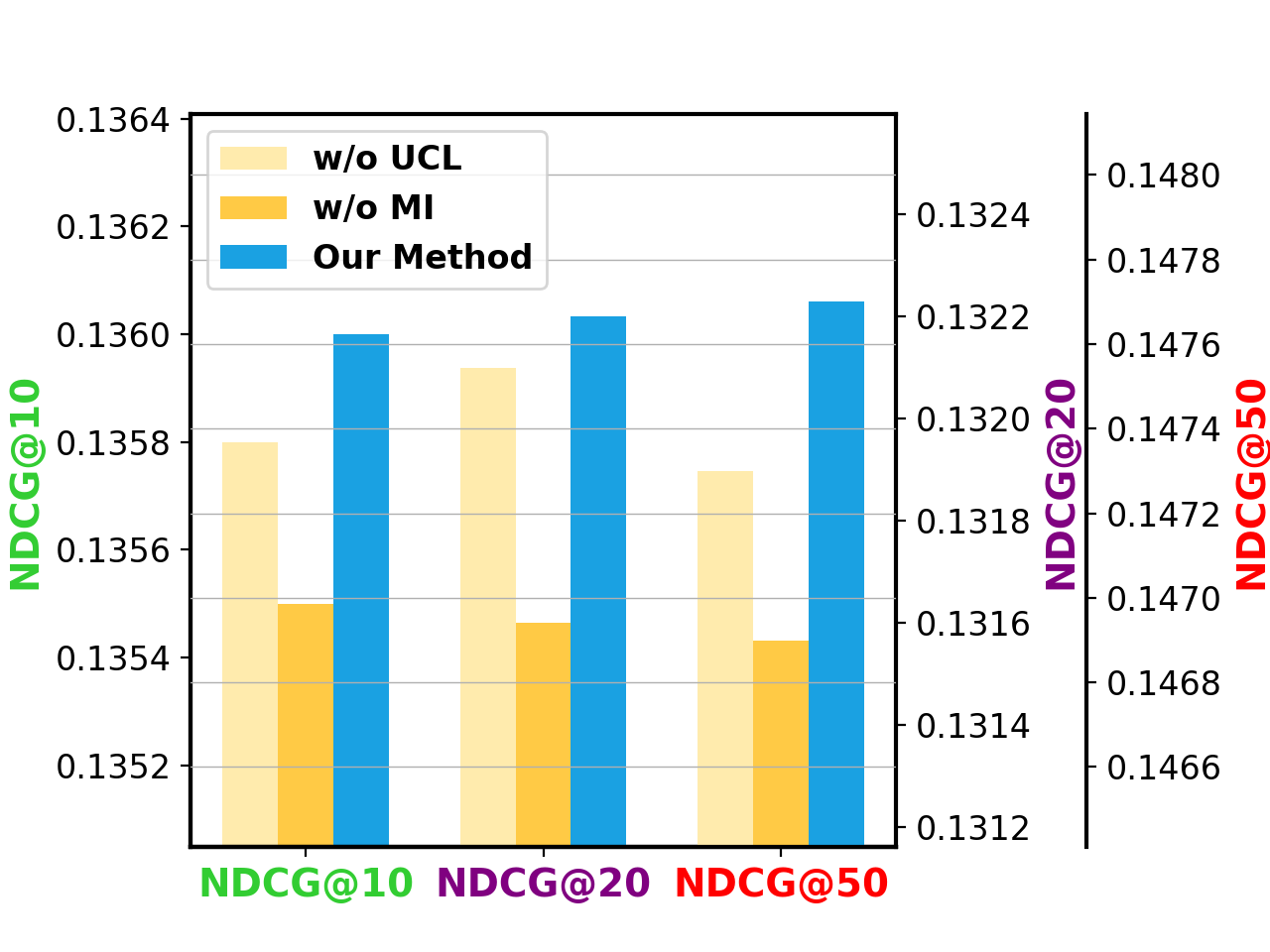}}
\subfigure[Ciao - NDCG@K]{\label{fig:module_ablation_ciao_recall}\includegraphics[width=0.65\columnwidth]{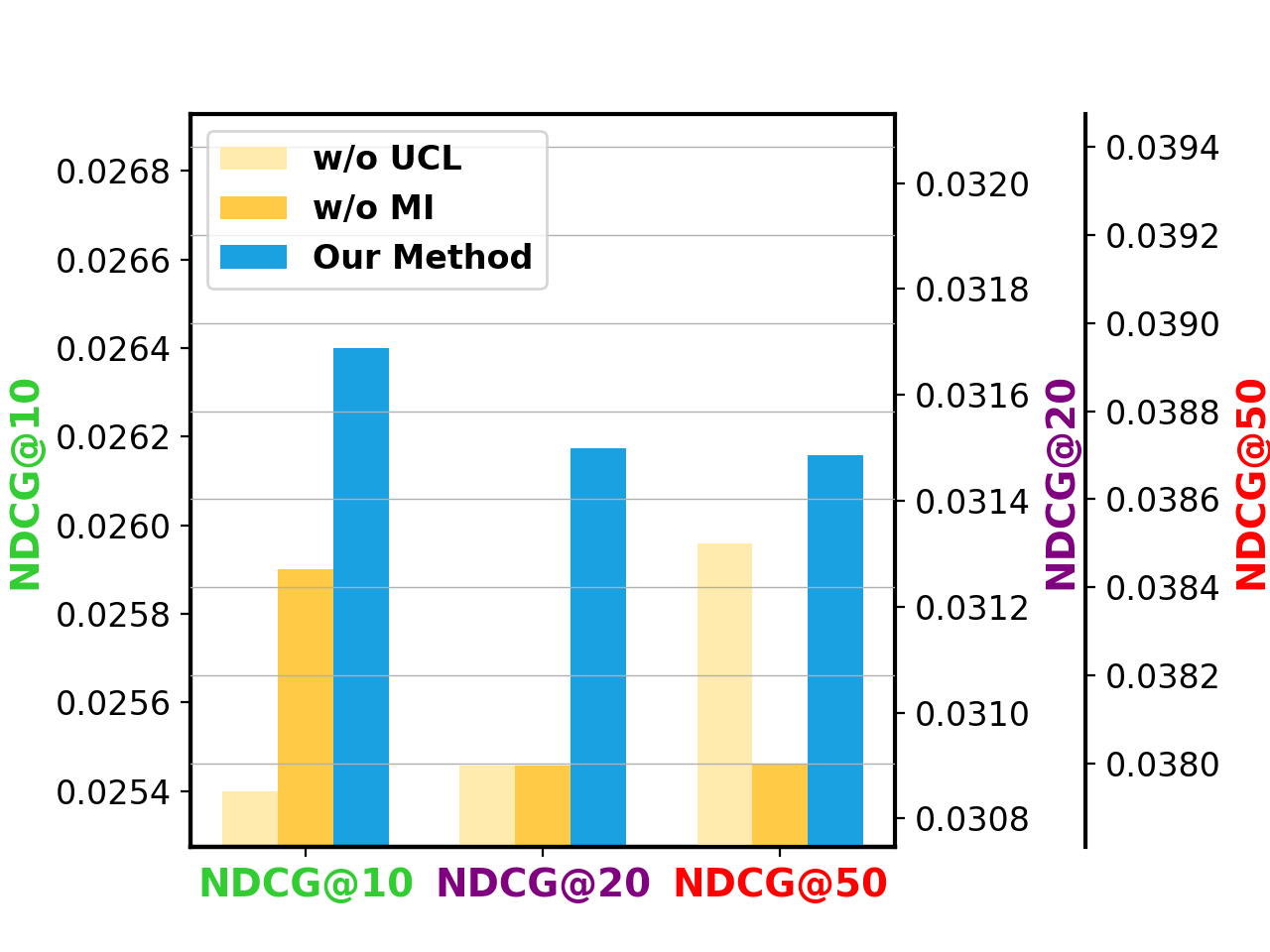}}
\subfigure[MovieLens-1M - NDCG@K]{\label{fig:module_ablation-1m_ndcg}\includegraphics[width=0.65\columnwidth]{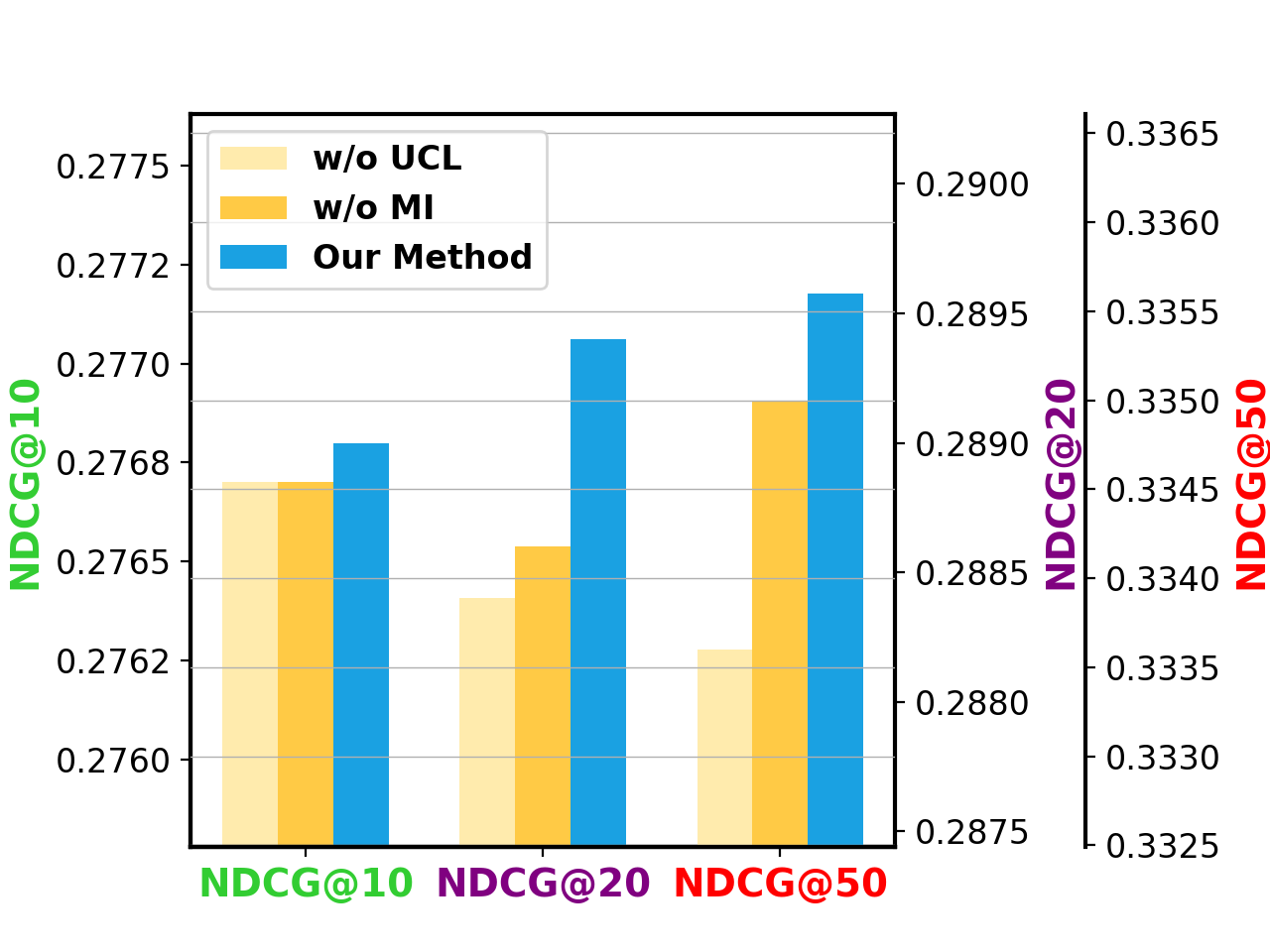}}
\vskip -0.15in
\caption{Ablation study on the effectiveness of different modules. \textit{w/o UCL} denotes the variant without the uniformly contrastive learning module and \textit{w/o MI} denotes the one without co-clustering module. This figure shows the results on datasets.}
\label{fig:performance-ablation-appendix}
\vskip -0.05in
\end{figure*}

\begin{figure*}[htbp]
\centering
\subfigure[Douban - Recall@K]{\label{fig:ssl_temp_douban_recall}\includegraphics[width=0.65\columnwidth]{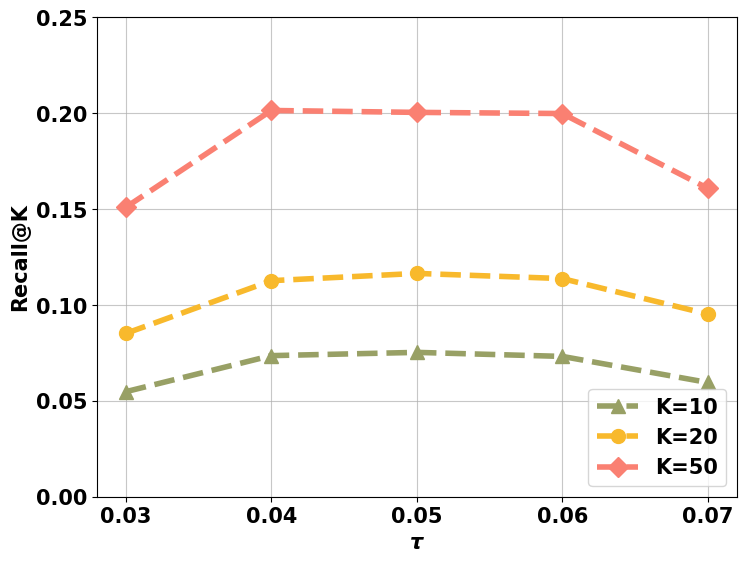}}
\subfigure[Ciao - Recall@K]{\label{fig:ssl_temp_ciao_recall}\includegraphics[width=0.65\columnwidth]{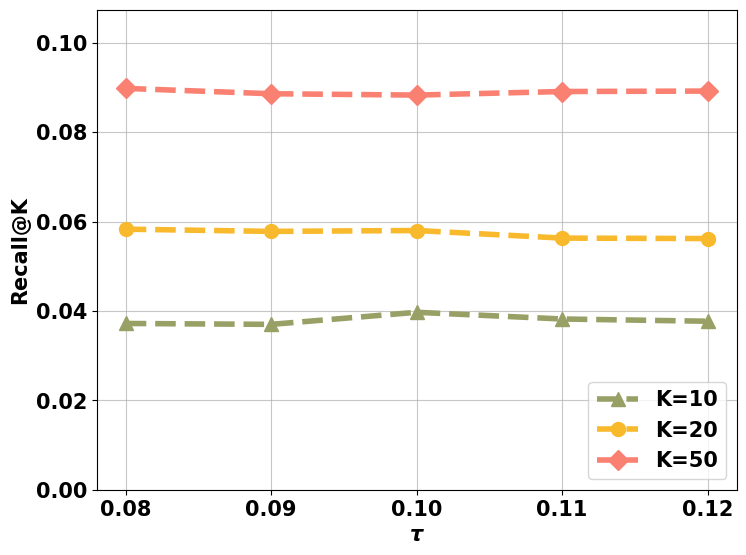}}
\subfigure[Amazon-Books - Recall@K]{\label{fig:ssl_temp_amazon-books_recall}\includegraphics[width=0.65\columnwidth]{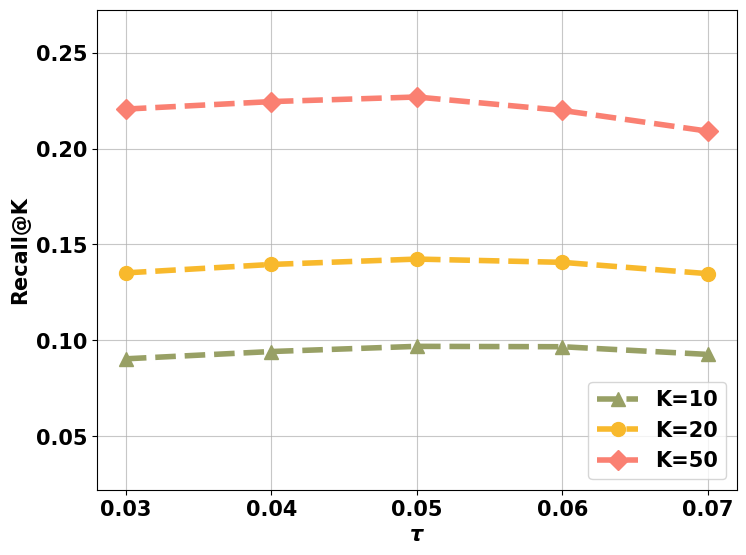}}

\vskip -0.1in
\caption{Performance comparison w.r.t. different values of temperature hyper-parameter $\tau$ on three datasets under metric Recall@K.}
\label{fig:hyperparameter-temperature-tau-appendix}
\vskip -0.1in
\end{figure*}

\begin{figure*}[htbp]
\centering

\subfigure[Ciao - $\lambda_1$]{\label{fig:lambda1-target-ssl_ciao_recall}\includegraphics[width=0.65\columnwidth]{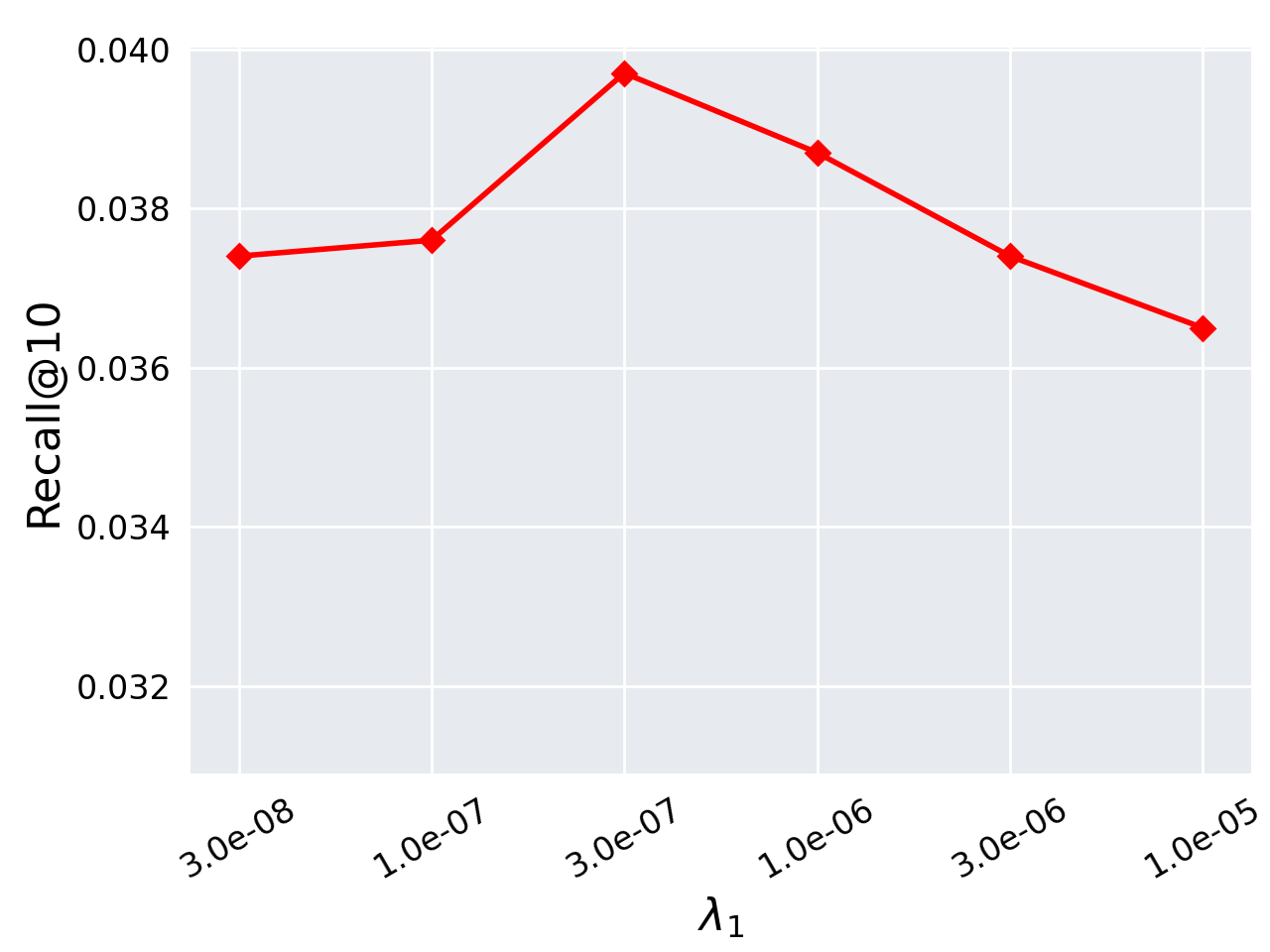}}
\subfigure[Ciao - $\lambda_2$]{\label{fig:lambda2-cluster-mi_ciao_recall}\includegraphics[width=0.65\columnwidth]{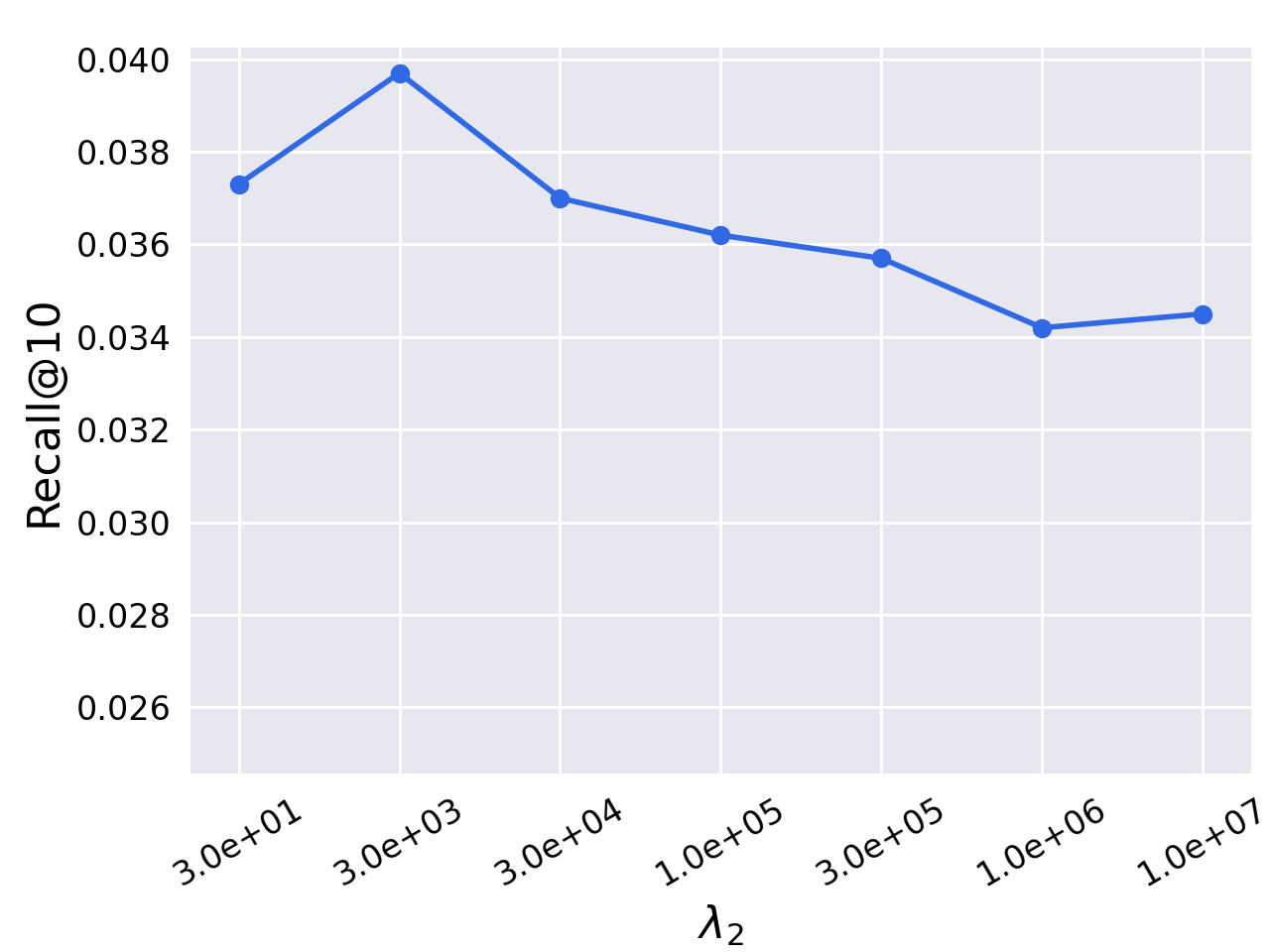}}
\subfigure[Ciao - $\lambda_3$]{\label{fig:lambda3-ssl-reg_ciao_recall}\includegraphics[width=0.65\columnwidth]{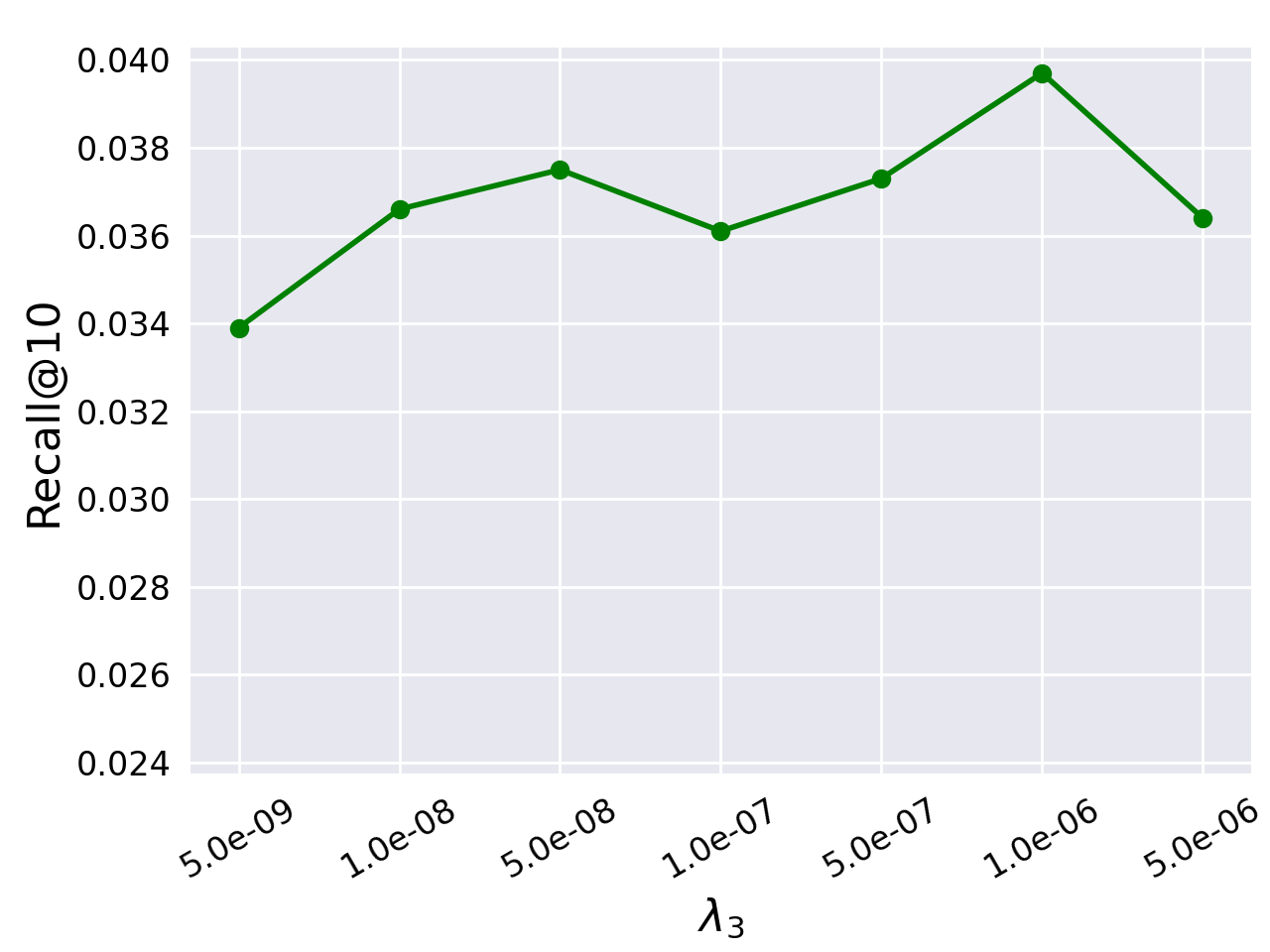}}

\subfigure[MovieLens-1M - $\lambda_1$]{\label{fig:lambda1-target-ssl-1m_recall}\includegraphics[width=0.65\columnwidth]{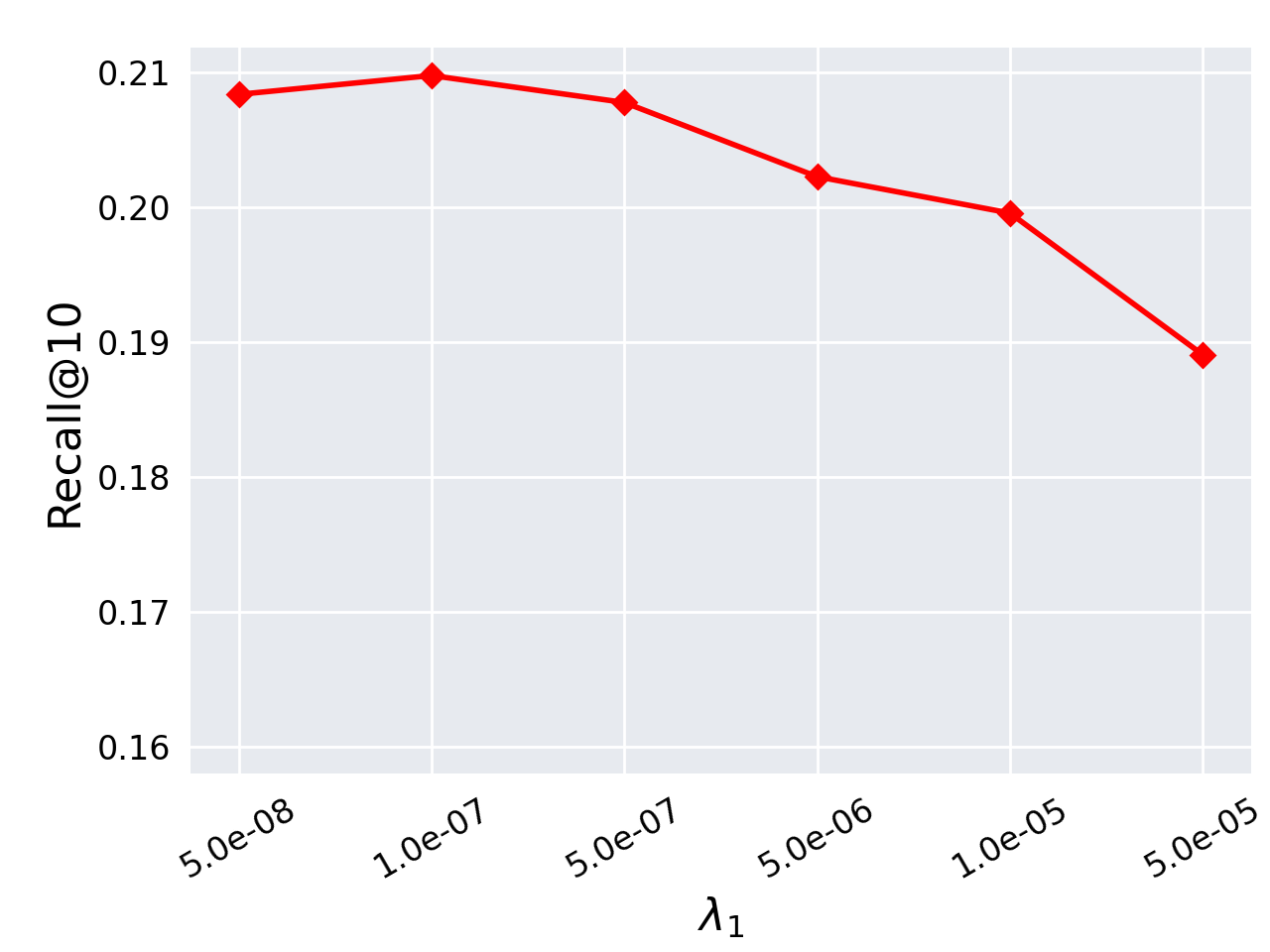}}
\subfigure[MovieLens-1M - $\lambda_2$]{\label{fig:lambda2-cluster-mi-1m_recall}\includegraphics[width=0.65\columnwidth]{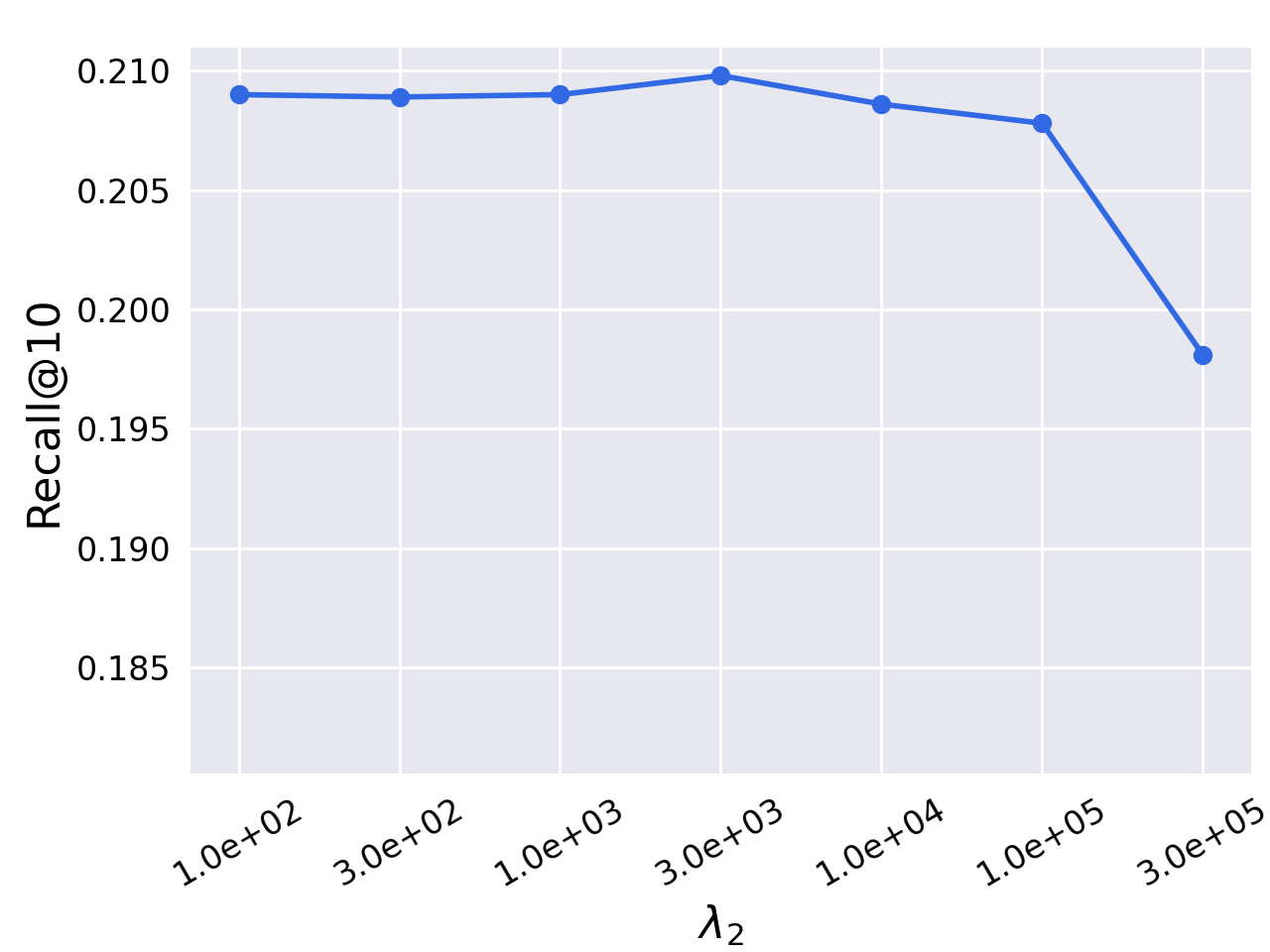}}
\subfigure[MovieLens-1M - $\lambda_3$]{\label{fig:lambda3-ssl-reg-1m_recall}\includegraphics[width=0.65\columnwidth]{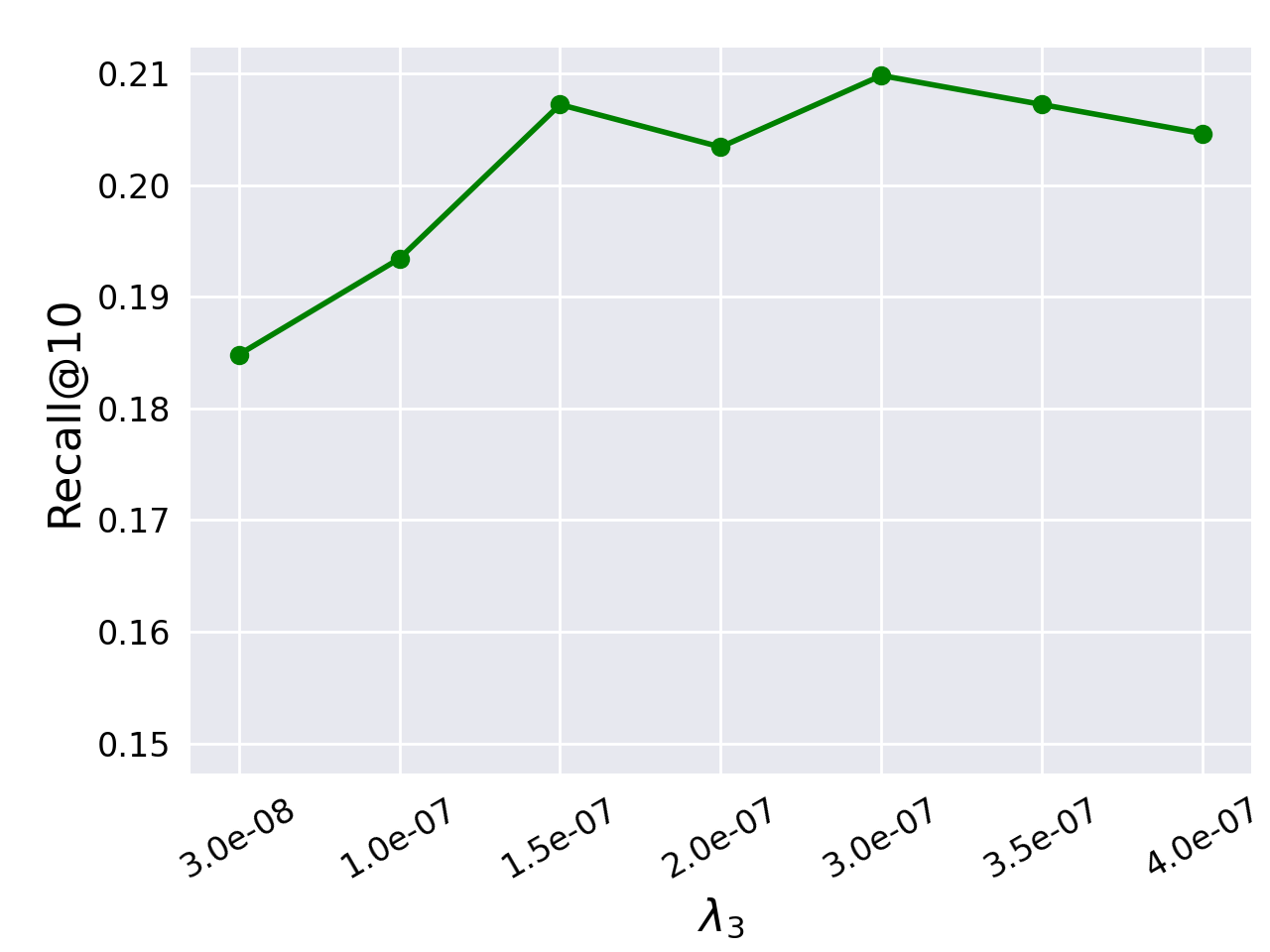}}

\subfigure[Amazon-Books - $\lambda_1$]{\label{fig:lambda1-target-ssl_amazon-books_recall}\includegraphics[width=0.65\columnwidth]{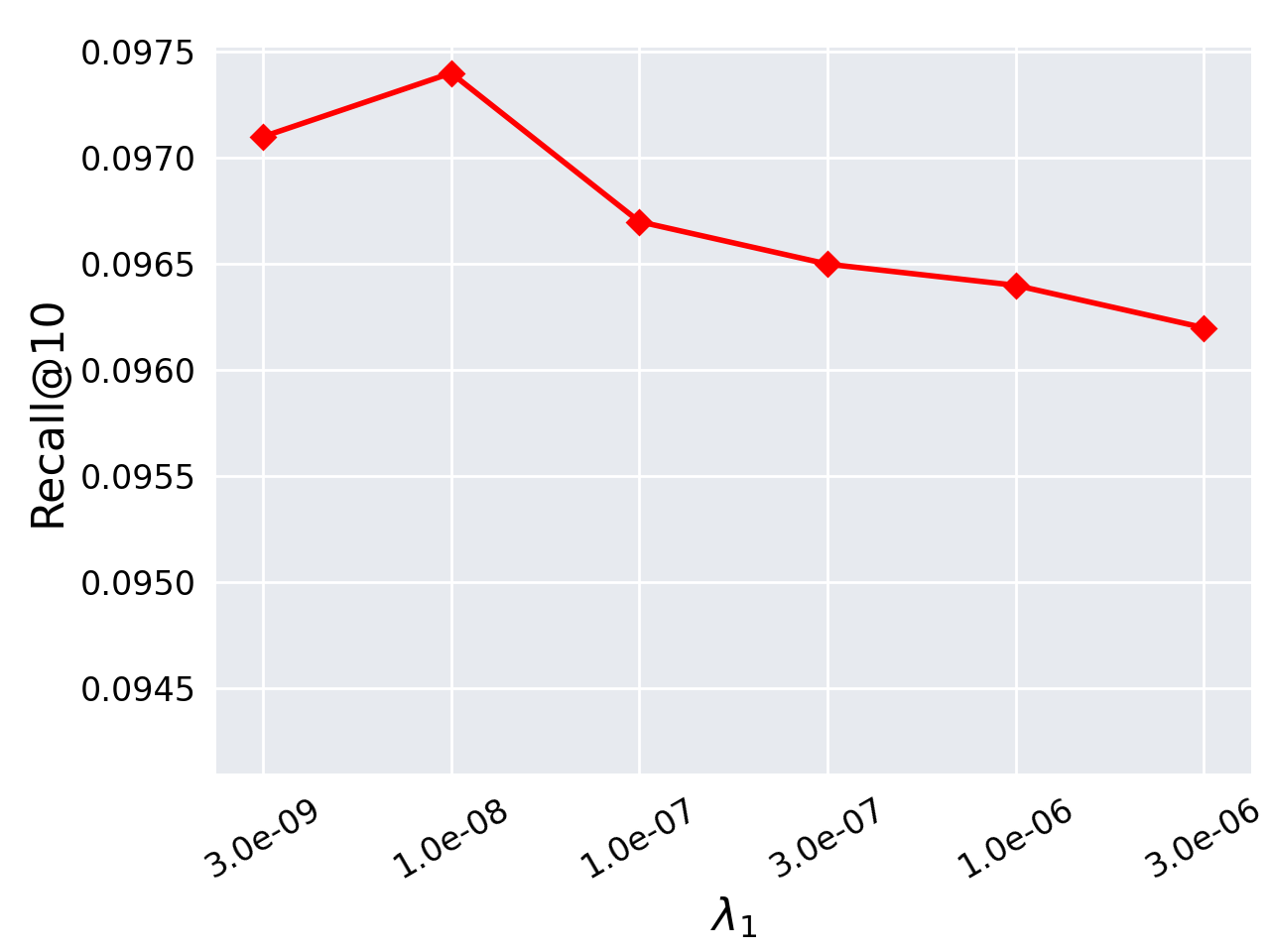}}
\subfigure[Amazon-Books - $\lambda_2$]{\label{fig:lambda2-cluster-mi_amazon-books_recall}\includegraphics[width=0.65\columnwidth]{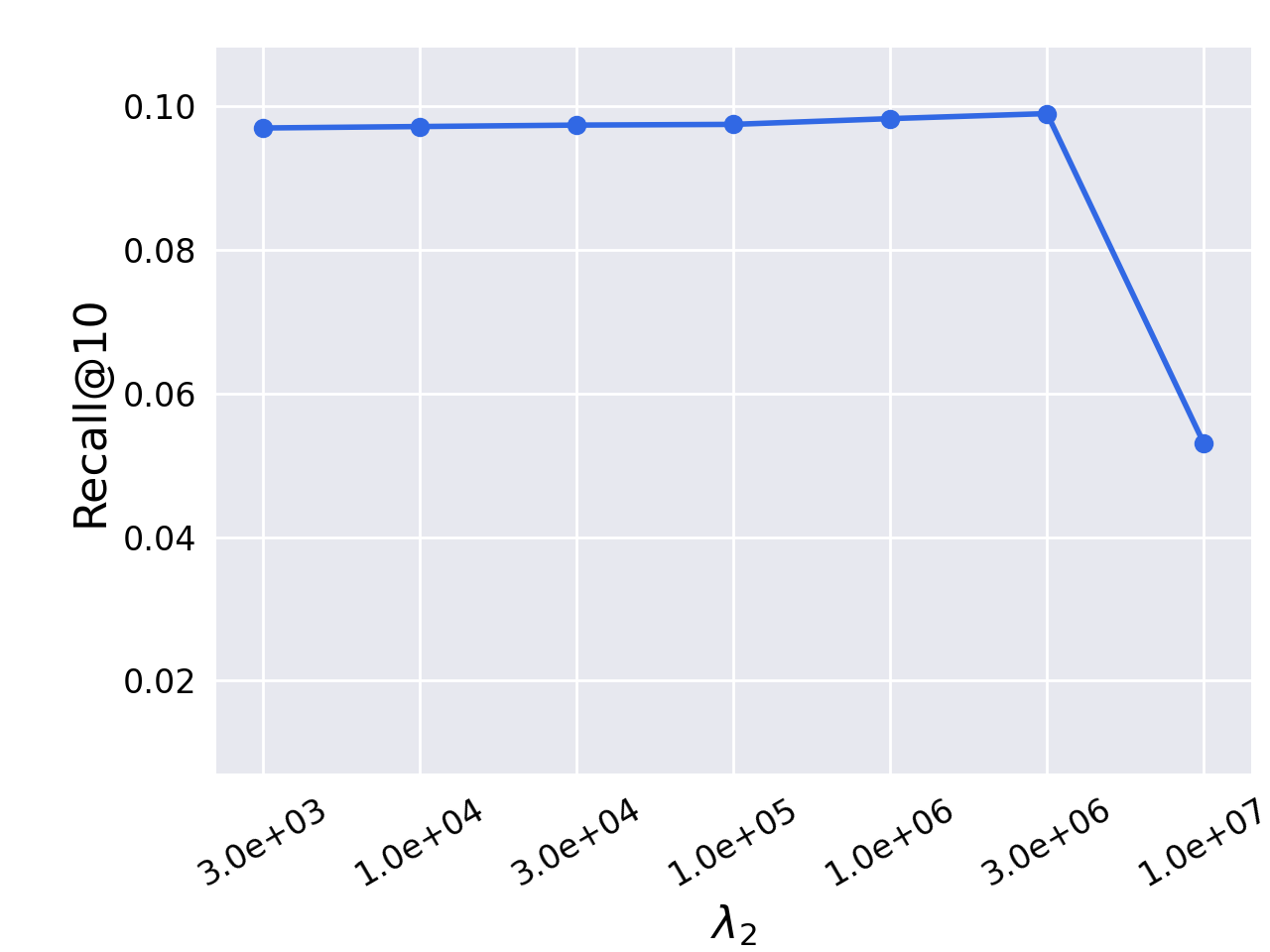}}
\subfigure[Amazon-Books - $\lambda_3$]{\label{fig:lambda3-ssl-reg_amazon-books_recall}\includegraphics[width=0.65\columnwidth]{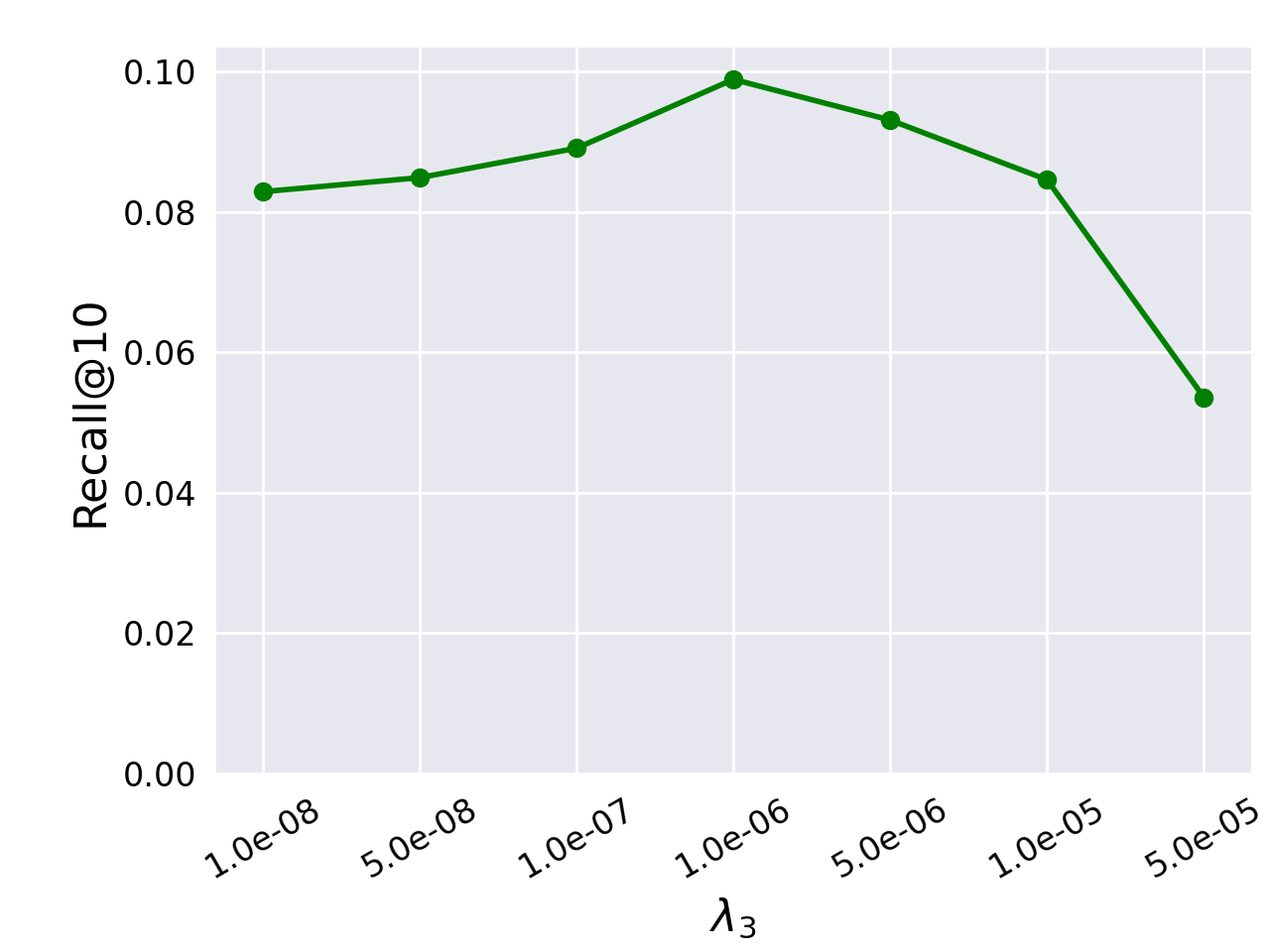}}
\vskip -0.15in
\caption{The effect of hyper-parameter $\lambda_1$, $\lambda_2$ and $\lambda_3$ under metric Recall@10 on three datasets.}
\label{fig:hyperparameter-lambdas-appendix}
\vskip -0.15in
\end{figure*}

\end{document}